\documentclass[conference,letterpaper]{IEEEtran}
\IEEEoverridecommandlockouts
\usepackage{mathtools,amssymb,lipsum, nccmath}

\usepackage{cuted}
\usepackage{lipsum}
\usepackage{mathtools}
\usepackage{cuted}

\usepackage{graphicx}
\usepackage{verbatim}
\usepackage{cite}
\usepackage{amsmath,amssymb,amsfonts,amsthm,stmaryrd}
\usepackage{algorithmic}
\usepackage{graphicx}
\usepackage{tikz}
\usepackage{graphicx}
\usepackage{bbm}
\usepackage{bm}
\usepackage{textcomp}
\usepackage{xcolor}
\usepackage{graphicx}
\usepackage{pgfplots}
\usepackage{cite}
\usepackage{systeme}

\usetikzlibrary{patterns}
\usepackage{bigints}
\usepackage{color}
\usepgfplotslibrary{fillbetween}

\usepackage{soul}
\usepackage{bm}
\usepackage{textcomp}
\usepackage{xcolor}
\usepackage{subcaption} 

\DeclareMathOperator*{\argmin}{arg\,min}
\def\BibTeX{{\rm B\kern-.05em{\sc i\kern-.025em b}\kern-.08em
    T\kern-.1667em\lower.7ex\hbox{E}\kern-.125emX}}

\newtheorem{theorem}{Theorem}
\newtheorem{remark}{Remark}
\newtheorem{definition}{Definition}

\newtheorem{lemma}{Lemma}

\newcommand{\N}{\mathbb{N}}

\newcommand{\R}{\mathbb{R}}
\newcommand{\E}{\mathbb{E}}

\usepackage[utf8]{inputenc} 
\usepackage[T1]{fontenc}
\usepackage{url}
\usepackage{ifthen}
\usepackage{cite}

\begin{document}

\title{Strategic Communication via Cascade Multiple-Description Network}

% %%% Single author, or several authors with same affiliation:
 
 \author{\IEEEauthorblockN{Rony Bou Rouphael and Ma\"{e}l Le Treust%\IEEEauthorrefmark{1} 
\thanks{%\IEEEauthorrefmark{1}
Ma\"el Le Treust gratefully acknowledges financial support from INS2I CNRS, DIM-RFSI, SRV ENSEA, UFR-ST UCP, INEX Paris Seine Initiative and IEA Cergy-Pontoise. This research has been conducted as part of the project Labex MME-DII (ANR11-LBX-0023-01).} 
}\\
\IEEEauthorblockA{
ETIS UMR 8051, CY Cergy-Paris Université, ENSEA, CNRS,\\
6, avenue du Ponceau, 95014 Cergy-Pontoise CEDEX, FRANCE\\
Email: \{rony.bou-rouphael ; mael.le-treust\}@ensea.fr}\\
 }

\maketitle

%%%%%%
%% Abstract: 
%%
\begin{abstract}
In decentralized decision-making problems, agents choose their actions based on locally available information and knowledge about decision rules or strategies of other agents.   
We consider a three-node cascade network with an encoder, a relay and a decoder, having distinct objectives captured by cost functions. In such a cascade network, agents choose their respective strategies sequentially, as a response to the former agent's strategy and in a way to influence the decision of the latter agent in the network.
We assume the encoder commits to a strategy before the communication takes place. Upon revelation of the encoding strategy, the relay commits to a strategy and reveals it. The communication starts, the source sequence is drawn and processed by the encoder and relay. Then, the decoder observes a sequences of symbols, updates its Bayesian posterior beliefs accordingly, and takes the optimal action. This is an extension of the Bayesian persuasion problem in the Game Theory literature. 
In this work, we provide an information-theoretic approach to study the fundamental limit of the strategic communication via three-node cascade network. 
Our goal is to characterize the optimal strategies of the encoder, the relay and the decoder, and study the asymptotic behavior of the encoder's minimal long-run cost function.
%We study the information-theoretic limits of the Bayesian persuasion game between autonomous devices through a cascade multi-user network with successive encoders' commitments. At each stage, information is transmitted between one committed encoder and two decoders, one of which has access to the observation of the other. The communication channel is perfect between each encoder and its corresponding decoders at all stages. Intermediary devices are endowed with a pair of strategies each, one for encoding and another for decoding. Each player is equipped with a distinct and arbitrary cost function. We study the strategic source coding problem at each stage, in which the encoder commits to an encoding while its corresponding decoders select the sequences of symbols that minimize their respective long-run cost functions. We characterize the optimal encoder cost value by considering successive refinement coding with respect to a specific probability distribution which involves two auxiliary random variables, and captures the incentive constraints of both decoders.
\end{abstract}

%\textit{A full version of this paper is accessible at:}
%\url{https://arxiv.org/pdf/2105.06201.pdf} 

%\lipsum[1-2]

\section{Introduction} 

We study a decentralized decision-making problem with restricted communication between three agents with non-aligned objectives.  %optimality is not determined by the precision of the conveyed information, but it rather refers to attaining an objective subject to the challenges imposed by the channel's bandwidth.
As depicted in Fig. \ref{fig:Cascademulti}, we consider a Cascade channel where information travels from a strategic encoder to a decoder through a strategic relay. 
We are interested in designing an achievable multiple description coding scheme that minimizes the encoder's long run cost function subject to the challenges imposed by the Cascade channel. 

The problem of strategic communication originally emerged in the game theory literature to address situations in economics (lobbying, advertising, sales, negotiations, etc.).  %This problem, of game theoretic nature, has become a multi-disciplinary subject of study due to its numerous applications in several fields. 
%This paper extends our work in \cite{rouphael2021strategic} to the case where each decoder observes a private and a public signal, and cost functions depend on the actions of both decoders. \\ 
The game was referred to as the sender-receiver game, and communication was assumed to be perfect and unconstrained by any limits on the amount of information transmitted. The Nash equilibrium solution of the cheap talk game was investigated by Crawford and Sobel in their seminal paper 
\cite{crawford1982}, in which the encoder and the decoder are endowed with distinct objectives and choose their coding strategies simultaneously. The Stackelberg version of the strategic communication game, referred to as the Bayesian persuasion game, was formulated by Kamenica and Gentzkow in \cite{KamenicaGentzkow11}, where the encoder is the Stackelberg leader and the decoder is the Stackelberg follower choosing its strategies as a response to the encoder's strategy
. In this paper, we assume that the encoder commits to an encoding and announces its commitment before observing the source. Then, the relay commits to and announces a strategy accordingly. If the relay was assumed to commit to a strategy before the encoder, then the problem boils down to a strategic joint source-channel coding of Shannon, like the one investigated in \cite{jet}. 
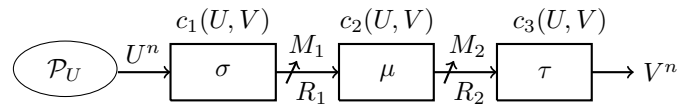
\begin{figure}

\centering
  \begin{tikzpicture}[xscale=2.8,yscale=1.5]
%\draw[thick,-](0,0)--(0.5,0)--(0.5,0.5)--(0,0.5)--(0,0);
\node[below, black] at (0.25,0.45) {$\mathcal{P}_U$};
\draw[thick,->](0.5,0.25)--(0.75,0.25);
\draw[thick,-](0.75,0)--(1.25,0)--(1.25,0.5)--(0.75,0.5)--(0.75,0);
\draw[thick,->](1.25,0.25)--(1.55,0.25);
\draw[thick,-](1.55,0)--(2,0)--(2,0.5)--(1.55,0.5)--(1.55,0);
\draw[thick,-](2.3,0)--(2.75,0)--(2.75,0.5)--(2.3,0.5)--(2.3,0);
\draw[thick,->](2,0.25)--(2.3,0.25);
\node[below, black] at (1,0.4) {$\sigma$};
\node[below, black] at (1.78,0.4) {$\mu$};
%\draw[thick,-](2,0.25)--(2.10,0.25)--(2.10,0.85);
\draw (0.25,0.28) circle (7pt);
%\draw (1.79,0.28) circle (7pt);
\draw[thick, ->](2.75,0.25)--(2.95,0.25);
%\draw[thick,-](2,0.25)--(2.10,0.25)--(2.10,-0.45);
%\draw[dashed, ->](2.65,-0.45)--(2.65,0);
%\draw[thick,-](2.25,0.45)--(2.55,0.45)--(2.55,1.25)--(2.25,1.25)--(2.25,0.45);
%\node[below, black] at (2.4,1.10) {$c_1$};
%\draw[thick,-](2.25,-0.75)--(2.55,-0.75)--(2.55,0)--(2.25,0)--(2.25,-0.75);
%\node[right, black] at (2.65,-0.25) {$Z^n$};
\draw[thick,->](1.3,0.15)--(1.35,0.35);
\draw[thick,->](2.05,0.15)--(2.095,0.35);
\node[right, black] at (2.45,0.25) {$\tau$};
\node[above, black] at (1.78,0.5) {$c_2(U,V)$};
\node[above, black] at (2.6,0.5) {$c_3(U,V)$};
\node[above, black] at (1,0.5) {$c_1(U,V)$};
\node[right, black] at (2.95,0.25) {$V^n$};
\node[above, black] at (0.62,0.25) {$U^n$};
\node[above, black] at (1.4,0.3) {$M_1$};
\node[above, black] at (2.16,0.3) {$M_2$};
\node[below, black] at (1.42,0.25) {$R_1$};
\node[below, black] at (2.17,0.25) {$R_2$};
    \end{tikzpicture} 
 \caption{Strategic Source Coding for Cascade Channel with Successive Commitment. %XXX add the cost functions to the figure
}
    \label{fig:Cascademulti}
\end{figure}
Cascade source coding consists of compressing a source sequence through an intermediate or relay node which then reconstructs the source and transmits it to the next node. In \cite{cascade1}, Yamamoto considered the source coding problem for cascade and branching communication systems, and established the region of achievable rates for cascade systems and bounds for the branching systems. Lossy source coding for cascade communication systems was also considered in \cite{cascade2} where both the relay and the terminal node have access to side information and wish to reconstruct the source with certain fidelities. %The cascade source coding framework is a building block for various compression and communication scenarios. It captures major aspects of multihop coding for wireless communication networks \cite{cascade3}, including cellular communication \cite{cascade4}, ad hoc networks \cite{cascade5}, and sensor networks \cite{cascade6}. 

In Bayesian persuasion, the encoder is considered to be an information designer. In \cite{koesslerlaclau}, information design with multiple designers interacting with a set of agents is studied.   
In \cite{sar1}, \cite{SaritasFurrerGeziciLinderYukselISIT2019}, the Nash equilibrium solution is investigated for multi-dimensional sources and quadratic cost functions, whereas the Stackelberg solution is studied in \cite{sar2}. The computational aspects of the persuasion game are considered in \cite{dughmi}. 
In several recent contributions \cite{voraachievable}, \cite{vorakulkarniinformationextraction}, \cite{vorakulkarni2021optimalquestionnaires}, Vora and Kulkarni addressed the problem of extracting truthful information from a strategic sender with an incentive to misreport information.  Modeled as a Stackelberg game in which the decoder is the Stackelberg leader, authors investigate the region of achievable rates of strategic communication between agents with distinct utility functions. 
The strategic communication problem with a noisy channel is investigated in \cite{AkyolLangbortBasar15}, \cite{akyol2017information}, \cite{LeTreustTomala(Allerton)16}, \cite{jet}, and %four different scenarios of strategic communication are studied in 
\cite{pointtopoint}. The case where the decoder privately observes a signal correlated to the state, also referred to as the Wyner-Ziv setting \cite{wyner-it-1976}, is studied in \cite{akyol2016role}, \cite{corsica2020} and \cite{LeTreustTomalaISIT21}.

In this paper, we study the Bayesian persuasion game via a Cascade multiple description network. %The encoder $\mathcal{E}$ commits to and reveals an encoding strategy before observing the source. Each commitment of the encoder induces a Bayesian game among the decoders $\mathcal{D}_1$ and $\mathcal{D}_2$. Since the set of information policies is compact, the Bayesian game admits Bayes-Nash equilibria \cite{Yu1999EssentialEO}. 
%We assume that decoders will select the pair of output sequences that minimizes their respective costs and maximizes the encoder's cost.
%update their Bayesian beliefs about the source sequence and select the output sequence that minimizes their respective cost functions. 
The objectives of the players are captures by distinct cost functions that depend on the source and the action taken by the decoder.
For some particular cases, we are able to characterize the encoder's optimal cost obtained with strategic cascade multiple description coding %with respect to the distribution that involves three auxiliary random variables, and 
that satisfies the decoder's incentives constraints.

\subsection{Notations}

%\subsubsection{Notations}

%Let $\mathcal{E}$ denote the encoder and $\mathcal{D}_i$ denote the decoder $i \in \{1,2\}$. 
Let $n \in \N^{\star}=\mathbb{N}\backslash\{0\}$ denote a sequence block length.
We denote by $U^n$ the $n$-sequences of random variables of source information $u^n=(u_{1},...,u_{n}) \in \mathcal{U}^n $, and by $V^n$ the sequences of decoders' %$\mathcal{D}_i$ 
actions $v^n \in \mathcal{V}^n$. Sequences $M_1$ and $M_2$ denote the channel inputs of encoders $1$ and $2$ respectively.  Calligraphic fonts $\mathcal{U}$, %$\mathcal{X}_1$, $\mathcal{X}_2$ 
 and $\mathcal{V}$ denote the alphabets and lowercase letters $u$ 
 and $v$ denote the realizations. %We assume that $|\mathcal{U}|=|\mathcal{X}_1|=|\mathcal{X}_2|$.
 %%such that $\mathcal{P}_{U^n,Z^n_1,Z_2^n} %= \mathcal{P}(U^n)\mathcal{P}(Z_1^n|U^n)\mathcal{P}(Z_2^n|Z_1,U^n).$
%The private observations $Z_1$ of decoder $\mathcal{D}_1$ and $Z_2$ of decoder $\mathcal{D}_2$ are correlated to $(U_0,U_1,U_2)$ according to the conditional probability distribution  $\mathcal{P}_{Z_1,Z_2|U}$.
%The i.i.d. memory-less channel distribution will be denoted by $\mathcal{T}_{Y_1,Y_2|X}$ such that $\mathcal{T}_{Y_1^n,Y_2^n|X^n}(Y_1^n,Y_2^n|X^n)=\prod_{t=1}^n \mathcal{T}_{Y_{1,t},Y_{2,t}|X_t}(Y_{1,t},Y_{2,t}|X_t).$\\
%Posterior beliefs of decoder $\mathcal{D}_i$ follow the probability distribution $\mathcal{Q}_i(u|z_i)$ over $\ \mathcal{U}\times\mathcal{Z}_i$. The joint probability distribution $\mathcal{Q}_i(u,z_1,z_2) \in \Delta(\mathcal{U}\times\mathcal{Z}_1\times\mathcal{Z}_2)$ decomposes as follows: $\mathcal{Q}_i(U^n,Z_1^n,Z_2^n) = \mathcal{Q}_i(U^n)\mathcal{Q}_i(Z_1^n|U^n)\mathcal{Q}_i(Z_2^n|Z_1,U^n).$\\
%Posterior beliefs of decoder $\mathcal{D}_1$ follow the probability distribution $\mathcal{Q}^1_{U^n|M_1M_2} \in \Delta(\mathcal{U}^n)$, and beliefs of decoder $\mathcal{D}_2$ follow the probability distribution $\mathcal{Q}^2_{U^n|M_2} \in \Delta(\mathcal{U}^n)$. 
\begin{comment}
The joint probability distribution $\mathcal{Q}_{U^n,W^n_1,W^n_2} \in \Delta(\mathcal{U}^n\times\mathcal{W}^n_1\times\mathcal{W}^n_2)$ decomposes as follows $\mathcal{Q}^i_{U^n,W_1^n,W_2^n} = \mathcal{Q}^i_{U^n}\mathcal{Q}^i_{W_1^nW_2^n|U^n}\mathcal{Q}^i_{W_2^n|Y^n_1,U^n}.$ \\
\end{comment}
For a discrete random variable $X,$ we denote by $\Delta(\mathcal{X})$ the probability simplex, i.e. the set of probability distributions over $\mathcal{X},$ and by $\mathcal{P}_{X}(x)$ the probability mass function $\mathbb{P}\{X=x\}$. Notation $X -\!\!\!\!\minuso\!\!\!\!- Y -\!\!\!\!\minuso\!\!\!\!- Z$ stands for the Markov chain property $\mathcal{P}_{Z|XY}=\mathcal{P}_{Z|Y}$.

% \mlt{do not use "small", remove everywhere}
 \section{System Model}
In this section, we formulate the coding problem. Let $n\in\mathbb{N}^{\star}$, and $(R_1,R_2) \in \mathbb{R}^2_{+}$ denote the rate pair. %Consider alphabets $\mathcal{X}_1$ and $\mathcal{X}_2$ such that $ |\{1,..2^{\lfloor nR_1\rfloor}\}|=2^{nR_1}$ and $|\{1,..2^{\lfloor nR_2\rfloor}\}|=2^{nR_2}$. 
We assume that the information source $U$ follows the independent and identically distributed (i.i.d) probability distribution $\mathcal{P}_{U}\in\Delta(\mathcal{U})$.
\begin{definition}
The coding strategies $\sigma$, $\mu$ and $\tau$ of the encoder, relay, and decoder respectively are defined by
\begin{align}
    \sigma:& \ \mathcal{U}^n \longrightarrow \Delta\Big(\{1,..2^{\lfloor nR_1\rfloor}\}\Big),\\
     \mu:& \{1,..2^{\lfloor nR_1\rfloor}\} \longrightarrow \Delta\Big(\{1,..2^{\lfloor nR_2\rfloor}\}\Big),\\
      \tau:& \{1,..2^{\lfloor nR_2\rfloor}\} \longrightarrow \Delta\Big(\mathcal{V}^n\Big).
%\tau_2:& \{1,2,..2^{\lfloor nR_2\rfloor}\} %%\times \mathcal{Z}_2^n 
%\longrightarrow \Delta(\mathcal{V}_2^n),
 \end{align} 
%We denote by $\mathcal{S}^0(n)$, $\mathcal{S}^1(n)$, and $\mathcal{S}^2(n)$ the sets of strategies $\sigma$, $\mu$ and $\tau$ respectively, and by $\mathcal{S}(n)$ the set of coding triplet ($\sigma,\mu,\tau)$.
\end{definition}
The stochastic coding strategies  ($\sigma,\mu,\tau)$
induce a joint probability distribution $\mathcal{P}^{\sigma\mu\tau} \in \Delta \big(\mathcal{U}^n \times \{1,2,..2^{\lfloor nR_1\rfloor}\}\times\{1,2,..2^{\lfloor nR_2\rfloor}\} 
 \times \mathcal{V}^n
\big)$ defined for all $(u^n,m_1,m_2,v^n)$  by
\begin{align}
 &\mathcal{P}^{\sigma\mu\tau}(u^n,m_1,m_2,v^n) = \nonumber\\
&\bigg(\prod_{t=1}^n\mathcal{P}_{U}(u_t)\bigg)\sigma(m_1|u^n)\mu(m_2|m_1)\tau(v^n|m_2). \label{jointproba123}
\end{align}

\begin{definition} \label{def:singlelongrun} 
We consider arbitrary single-letter cost functions $c_1:\mathcal{U}\times\mathcal{V} \longrightarrow \mathbb{R}$ for the encoder $\mathcal{E}$, $c_2:\mathcal{U}\times\mathcal{V} \longrightarrow \mathbb{R}$ for the relay, and $c_3:\mathcal{U}\times\mathcal{V} \longrightarrow \mathbb{R}$ for the decoder.
%of the encoder $c_e$ and both decoders $c_i$ for $\ i \in \{1,2\}$ are defined by
% \begin{align}
%      &c_e:\mathcal{U}\times\mathcal{V}_1\times\mathcal{V}_2 \longrightarrow \mathbb{R},\\
%      &c_1:\mathcal{U}\times\mathcal{V}_1 \longrightarrow \mathbb{R}, \\
%      &c_2:\mathcal{U}\times\mathcal{V}_2 \longrightarrow \mathbb{R}.
% \end{align}
 %\end{definition} 
%\mlt{never use "footnotesize" or "tiny" or other along the paper}
%\normalfont \mlt{remove the commands "normalfont" after you modify the definition environments evreywhere}
%\begin{definition}\label{def:longrun}
The long-run cost functions are defined by  
\begin{align*}
c_1^n(\sigma,\mu,\tau)=& \mathbb{E}_{\sigma,\mu,\tau}\Bigg[\frac{1}{n}\sum_{t=1}^n c_1(U_{t},V_t)\Bigg]  \\
=&\sum_{u^n, v^n} \mathcal{P}_{U^nV^n}^{\sigma,\mu,\tau}(u^n, v^n)\cdot\Bigg[\frac{1}{n}\sum_{t=1}^n c_1(u_{t},v_{t})\Bigg],  \\
c_2^n(\sigma,\mu,\tau)=& \mathbb{E}_{\sigma,\mu,\tau}\Bigg[\frac{1}{n}\sum_{t=1}^n c_2(U_{t},V_t)\Bigg]  \\
=&\sum_{u^n, v^n} \mathcal{P}_{U^nV^n}^{\sigma,\mu,\tau}(u^n, v^n)\cdot\Bigg[\frac{1}{n}\sum_{t=1}^n c_2(u_{t},v_{t})\Bigg],  \\
c_3^n(\sigma,\mu,\tau)=& \mathbb{E}_{\sigma,\mu,\tau}\Bigg[\frac{1}{n}\sum_{t=1}^n c_3(U_{t},V_t)\Bigg]  \\
=&\sum_{u^n, v^n} \mathcal{P}_{U^nV^n}^{\sigma,\mu,\tau}(u^n, v^n)\cdot\Bigg[\frac{1}{n}\sum_{t=1}^n c_3(u_{t},v_{t})\Bigg],  \\
%&c^n_{1}(\sigma,\tau_1)= %\mathbb{E}_{\sigma,\tau_1}\Bigg[\frac{1}{n}\sum_{t=1}^n c_1(U_t,V_{1,t})\Bigg]
%\sum_{u^n, v_1^n} \mathcal{P}_{U^nV_1^n}^{\sigma,\tau_1}(u^n, v_1^n)\cdot\Bigg[\frac{1}{n}\sum_{t=1}^n %c_1(u_t,v_{1,t})\Bigg], \nonumber \\ 
%&c^n_{2}(\sigma,\tau_2)=\sum_{u^n, v_2^n} \mathcal{P}_{U^nV_2^n}^{\sigma,\tau_2}(u^n, v_2^n)\cdot\Bigg[\frac{1}{n}\sum_{t=1}^n c_2(u_t,v_{2,t})\Bigg].
\label{dn2} 
\end{align*}
\end{definition} 
 In the above equations, $\mathcal{P}^{\sigma\mu\tau}_{U^nV^n}$ denote the marginal distributions over the sequences $(U^n,V^n)$ of $\mathcal{P}^{\sigma\mu\tau}$ defined in \eqref{jointproba123} over the $n$-sequences 
 $(U^n,M_1,M_2,V^n)$.  
 %-------
 %We consider the strategic communication game in which the encoder select $\sigma $ that minimizes $c_e^n(\sigma,\tau_1,\tau_2)$ and the decoders choose $\tau_i$, $i \in \{1,2\}$ in order to minimize the long-run cost functions ${c_i}^n(\sigma,\tau_i)$,  $i \in \{1,2\}$. 
\begin{definition} For any strategy pair $(\sigma, \mu)$, the set of best-response strategies $\tau%^{\star}(\sigma,\mu)
$ of the decoder is defined by
\begin{equation}
     \mathbb{A}_3(\sigma,\mu) = \argmin_{\tau}c_3^n(\sigma,\mu,\tau)% \in \mathcal{S}^2(n)
     .
   % \tau^{\star}(\sigma,\mu) = \argmin_{\tau}c_3^n(\sigma,\mu,\tau)
\end{equation}
For any strategy $\sigma$, %the optimal cost of the relay is given by
%\begin{align}
 %   C_2^{\star}(\sigma)=\inf_{\mu}c_2^n(\sigma,\mu,\sigma^{\star}_3(\sigma,\mu)).
%\end{align}
the set of best-response strategies $\mu%^{\star}(\sigma)
$ of the relay is defined by
\begin{equation}
   \mathbb{A}_2(\sigma)  %\mu^{\star}(\sigma) 
   = \argmin_{(\mu,\tau) s.t. \atop \tau \in\mathbb{A}_3(\sigma,\mu) }c_2^n(\sigma,\mu,\tau) %\in \mathcal{S}^2(n)\times\mathcal{S}^3(n)
   .
\end{equation}
%\begin{equation}
 %   BR_{3}(\sigma,\mu)=\Big\{\tau, \quad c_3^n(\sigma,\mu,\tau) \leq c_3^n(\sigma,\mu,\Tilde{\tau}), \quad \forall \ \Tilde{\tau} \Big\}. 
%\end{equation}
%For any strategy $\sigma$, the set of best-response strategies of player $2$ is defined by
%\begin{equation}
 %   BR_{2}(\sigma)=\Big\{\mu, \quad c_2^n(\sigma,\mu,\tau) \leq c_2^n(\sigma,\Tilde{\mu},\tau), \quad \forall \ \Tilde{\mu} \Big\}. 
%\end{equation}

\end{definition}

%If several pairs of best-response strategies $(\tau_1,\tau_2)\in BR_{1}(\sigma)\times BR_{2}(\sigma)$ are available, we assume that the worst pair $(\tau_1,\tau_2)$, from the encoder perspective, is selected. Therefore, the solution is robust to the exact specification of the decoding strategies. 
Therefore, the encoder has to solve the following coding problem,

\begin{equation}
\Gamma_{e}^n(R_1,R_2)=\underset{\sigma}{\inf}\underset{(\mu,\tau)\in \atop \mathbb{A}_2(\sigma)}{\max} c_1^n(\sigma,\mu,\tau). \label{LP123}
\end{equation} 
\begin{remark}
In order to get a robust solution concept, we assume that the encoder solves the problem for the worst case scenario, i.e. if more than one pair of strategies are available in $\mathbb{A}_2(\sigma)$, we consider the one that maximizes the encoder's cost.   
%Possible Scenarios for a given encoding strategy $\sigma$: \\
%1) The relay chooses $\mu$ and the decoder picks one action. 
%2) the relay chooses $\mu$ and the decoder is indifferent between 2 actions. In that case, the decoder selects one action randomly. 
\end{remark}
The operational significance of %game theoretic framework of our problem characterizes the operational significance of the encoder's long-runoptimal cost given in 
\eqref{LP123} corresponds to the persuasion game that is played in the following steps:   %The LP version of the encoder's long-runoptimal cost given in \eqref{LP}, recalls the game theoretic framework which characterizes its operational significance.  
%Listed below are the steps in which the persuasion game takes place:  
%We now discuss the operational significance of $C_e^n(R_1,R_2)$.
%In this framework, we assume that the strategic communication takes place as follows:
\begin{itemize}
\item The encoder chooses, announces the encoding $\sigma$.
\item knowing $\sigma$, the relay chooses, announces the encoding $\mu$.
\item Knowing $(\sigma,\mu)$, the 
decoder compute its best-response strategy $\tau$. %observes $M_2$ and draws %a sequence 
%$V^n$ according to %the strategy
%${\tau%^{\star}(\sigma,\mu)
%}_{V^n|M_2}$.
\item Sequences $U^n$ are drawn i.i.d with distribution $\mathcal{P}_{U}$.
\item Message sequence $M_1$ are encoded according to ${\sigma}_{M_1|U^n}$.
\item Message sequence $M_2$ are encoded according to ${\mu%^{\star}(\sigma)
}_{M_2|M_1}$.
\item The decoder observes $M_2$ and draws %a sequence 
$V^n$ according to %the strategy
${\tau%^{\star}(\sigma,\mu)
}_{V^n|M_2}$.
\item Cost functions $c_1^n(\sigma,\mu,\tau)$, $c_2^n(\sigma,\mu,\tau)$, $c_3^n(\sigma,\mu,\tau)$ are computed. 

\end{itemize}

%\vspace{-2cm}
 
%\vspace{-2cm}
\section{Cascade Multiple Description Coding} %\vspace{-0.4cm} %the relay is just the identity function relaying the message received from the sender. Q_{W_2|W_1}=Q_{W_2|W_2}

Consider the cooperative communication scenario where $c_1=c_2=c_3$, and all three agents share the objective of minimizing the same cost function. This setting corresponds to the standard coding setup of a cascade multiple description network \cite{elgamal}, under the assumption that the relay does not reconstruct the source, but only relays a message $M_2$, and the cost functions of the three players depend on the source and the decoder's action. 

Consider an auxiliary random variables $W \in \mathcal{W}$ such that $|\mathcal{W}|=|\mathcal{U}|$. %and  %Random variables $U,W_1,W_2,V$ satisfy the following Markov chains 
 %\begin{align*} 
  %   U   -\!\!\!\!\minuso\!\!\!\!-  W_1  -\!\!\!\!\minuso\!\!\!\!- W_2, \quad  
 %W_1 -\!\!\!\!\minuso\!\!\!\!-  W_2  -\!\!\!\!\minuso\!\!\!\!- V.
% \end{align*}
The set $\mathbb{Q}^c_0(R_1,R_2)$ 
 of target distributions is defined by:
    \begin{align}
%%{%\color{red}{
\mathbb{Q}^c_0(R_1,R_2) =& \{\mathcal{Q}_{W_2|U}; \  \  \min(R_1,R_2) \geq I(U;W_2) \}. %\subset \Delta(\mathcal{W}_1\times\mathcal{W}_2)^{|\mathcal{U}|},%}} 
   \label{q10123} %\\
%\mathbb{Q}_1 =& \{\mathcal{Q}_{W_2|W_1} \ s.t. \  R_2 \geq I(W_2;W_1)  \}\subset \Delta(\mathcal{W}_2)^{|\mathcal{W}_1|}.%}} 
 %  \label{q20123}
 \end{align}
 %The sets %$\mathbb{Q}^1_{1}(\mathcal{Q}_{W_1|U})$,
 %$\mathbb{Q}_{3}(\mathcal{Q}_{W_1|U},\mathcal{Q}_{W_2|U})$  and  $\mathbb{Q}_{2}(\mathcal{Q}_{W_1|U})$ %$\mathbb{Q}_{2}(\mathcal{Q}_{W_2|U})$ 
The single-letter best-response of the decoder is defined by:
 \begin{align} %\vspace{-0.4cm}
 & \mathbb{Q}^c_{3}(\mathcal{Q}_{W_2|U})= \argmin_{\mathcal{Q}_{V|W_2}}\mathbb{E}%_{\mathcal{P}_U\mathcal{Q}_{W_1|U}\mathcal{Q}_{W_2|W_1}\mathcal{Q}_{V|W_2}}
  [c_3(U,V)]. %\\
 %&\mathbb{Q}_{2}(\mathcal{Q}_{W_2|U})= \argmin_{\mathcal{Q}_{V|W_2}\in\mathcal{Q}_3(\mathcal{Q}_{W_2|U})}\mathbb{E}[c_2(U,V)],%\vspace{-0.4cm}
 \end{align} %\vspace{-0.4cm}
  The single-letter optimal cost $\Gamma_e^{c}(R_1,R_2)$ of the encoder is given by
 %\vspace{-0.4cm}
 \begin{align} 
\Gamma_e^{c}(R_1,R_2)=\underset{\mathcal{Q}_{W_2|U}\in\atop\mathbb{Q}^c_0(R_1,R_2)}{\inf}\underset{\mathcal{Q}_{V|W_2}\in\atop \mathbb{Q}^c_{3}(\mathcal{Q}_{W_2|U})}{\max}
 {\mathbb{E}}
 \Big[c_1(U,V) \Big]. 
 \end{align} %\vspace{-0.4cm}
\begin{theorem} \label{cmdn}%\vspace{-0.4cm} 
Let $(R_1,R_2) \in \R^2_{+}$. If $c_1=c_2=c_3$, then 
\begin{align} 
\lim_{n\longrightarrow \infty} \Gamma^n_e(R_1,R_2)= \underset{n\in\mathbb{N}^{\star}}{\inf} \Gamma^n_e(R_1,R_2) = \Gamma^{c}_e(R_1,R_2).
%\vspace{-0.4cm}
%&a) \forall \ \varepsilon>0, \ \exists \hat{n} \in \mathbb{N} \  \forall n \geq \hat{n}, \ \ \Gamma^n_e(R_1,R_2) \leq \Gamma_e^{c}(R_1,R_2) + \varepsilon. \nonumber %\vspace{-0.6cm}
%\\
%&b)\forall n \in \ \mathbb{N}, \hspace{2.8cm}\Gamma_e^n(R_1,R_2) \geq \Gamma_e^{c}(R_1,R_2) . \nonumber
\end{align} 
\end{theorem}
The proof of Theorem \ref{cmdn} can be directly derived from the proof of \cite[Theorem 20.4]{elgamal} by considering the relay's estimate to be a constant and its role is to only transition the message received from the encoder.

\section{Bayesian Persuasion with no Information Constraint} %\vspace{-0.5cm}
%\section{Characterizations}
We assume that the communication is perfect and unrestricted. Fix $\mathcal{Q}_{W_1|U}\mathcal{Q}_{W_2|W_1}$.
Consider two auxiliary random variables $W_1 \in \mathcal{W}_1$ and $W_2 \in \mathcal{W}_2$ such that $|\mathcal{W}_1|=|\mathcal{W}_2|=|\mathcal{U}|$ and  %Random variables $U,W_1,W_2,V$ satisfy the following Markov chains 
 \begin{align*} 
     U   -\!\!\!\!\minuso\!\!\!\!-  W_1  -\!\!\!\!\minuso\!\!\!\!- W_2, \quad  
 W_1 -\!\!\!\!\minuso\!\!\!\!-  W_2  -\!\!\!\!\minuso\!\!\!\!- V.
 \end{align*}

The single-letter best-responses are defined by:
 \begin{align} %\vspace{-0.2cm}
  \mathbb{Q}_{3}(\mathcal{Q}_{W_1|U},\mathcal{Q}_{W_2|W_1})=& \argmin_{\mathcal{Q}_{V|W_2}}\mathbb{E}%_{\mathcal{P}_U\mathcal{Q}_{W_1|U}\mathcal{Q}_{W_2|W_1}\mathcal{Q}_{V|W_2}}
  [c_3(U,V)], \nonumber %\vspace{-0.4cm}
  \\
 \mathbb{Q}_{2}(\mathcal{Q}_{W_1|U})=&\argmin_{(\mathcal{Q}_{W_2|W_1},\mathcal{Q}_{V|W_2}), \atop \mathcal{Q}_{V|W_2}\in\mathcal{Q}_3(\mathcal{Q}_{W_1|U},\mathcal{Q}_{W_2|W_1})}\mathbb{E}[c_2(U,V)], \nonumber
 \end{align} %\vspace{-0.4cm}
  The single-letter optimal cost $\Gamma_e$ of the encoder is given by %\vspace{-0.4cm}
 \begin{align}%\vspace{-0.4cm}
 \Gamma_e=\underset{\mathcal{Q}_{W_1|U}}{\inf}\underset{\mathcal{Q}_{W_2|W_1},\mathcal{Q}_{V|W_2} \atop \in \mathbb{Q}_{2}(\mathcal{Q}_{W_1|U})}{\max}
 {\mathbb{E}}_{\mathcal{P}_{U}\mathcal{Q}_{W_1|U}\atop \mathcal{Q}_{W_2|W_1}\mathcal{Q}_{V|W_2}} 
 \Big[ c_1(U,V) \Big]. \nonumber
\end{align} 
\theorem \label{2cmdn} If $R_1=R_2=\log|\mathcal{U}|$, then
\begin{align} %\vspace{-0.4cm}
\lim_{n\longrightarrow \infty} \Gamma^n_e= \underset{n\in\mathbb{N}^{\star}}{\inf} \Gamma^n_e = \Gamma_e
%&a) \qquad \forall \ \varepsilon>0, \ \exists \hat{n} \in \mathbb{N} \  \forall n \geq \hat{n}, \ \ \Gamma^n_e \leq \Gamma_e^{\infty}  + \varepsilon. \\
%&b)\qquad \forall n \in \ \mathbb{N},\hspace{2.8cm}\Gamma_e^n \geq \Gamma_e^{\infty}. \nonumber
\end{align}
%Using Fekete's Lemma for the sub-additive sequence $\big(n \Gamma_e^n \big)_{n\in \N^{\star}}$ %\cite[Lemma 1]{rouphael2021strategic}
%we get
%\begin{align} %\vspace{-0.5cm}
%\lim_{n\longrightarrow \infty} \Gamma^n_e= \underset{n\in\mathbb{N}^{\star}}{\inf} %\Gamma^n_e = \Gamma^{\infty}_e. %\vspace{-0.4cm}
%\nonumber \end{align}

%The expected values are evaluated with respect to $\mathcal{P}_U\mathcal{Q}_{W_1|U}\mathcal{Q}_{W_2|W_1}\mathcal{Q}_{V|%z_i,
% W_2}$. \\
%\vspace{-2cm}

\subsection{Achievability of Theorem \ref{cmdn}}

%Rate splitting: Divide the index $M_2$ into two independent indices $M_{02},M_{22}$ with respective rates $R_{02}$ and $R_{22}$, such that $R_2=R_{02}+R_{22}$. \\
Let $R_1=R_2=\log|\mathcal{U}|$, and %\in\mathbb{R}^2_{+}$
fix a joint probability distribution $\mathcal{Q}_{W_1|U}\mathcal{Q}_{W_2|W_1}%\in\mathbb{Q}^c_0(R_1,R_2)
$.  %and $\mathcal{Q}_{W_2|W_1}\in\mathbb{Q}_1$
%There exists $\eta>0$ such that \begin{align}
 %  \min(R_1,R_2) =& I(U;W_2) + \eta, 
    %R_2 =& I(U;W_2) + \eta. \label{eqwwee}
%\end{align} 
%$R_{22}=I(U;W_2)+\eta$ and $R_{02}=R_2-I(U;W_2)$ and $R_1+R_2 > I(U;W_1,W_2) + \eta$. \\
 %We define indices $(m_1,m_2) \in \mathcal{M}_1\times\mathcal{M}_2$ with $|\mathcal{M}_1|=2^{nR_1}%=|\mathcal{X}_1|
%$ and  $|\mathcal{M}_2|=2^{nR_2}%|\mathcal{X}_2|
%$. 
The sequences $U^n$ are drawn according to the i.i.d. distribution $\mathcal{P}_{U^n}$.
Randomly and independently generate $2^{nR_1}$ %%$|\mathcal{X}_1|$ 
sequences $w_1^n(m_1)$ for each $m_1 \in \{1,..2^{\lfloor nR_1\rfloor}\}$, 
according to the i.i.d distribution $\mathcal{Q}_{W_1^n|U^n}=\Pi_{t=1}^n\mathcal{Q}_{W_1|U}(w_{1t}|u_t)$. 
%For each pair $(m_1,m_2) \in \{1,..2^{\lfloor nR_1\rfloor}\}\times\{1,..2^{\lfloor nR_2\rfloor}\}$, randomly and independently generate a sequence $w_2^n(m_1,m_2)$ according to the i.i.d distribution $\mathcal{Q}_{W_2^n}=\Pi_{t=1}^n\mathcal{Q}_{W_2}(w_{2t})$. 
Similarly, generate $2^{nR_2}$ %$|\mathcal{X}_2|$
sequences $w_2^n(m_2)$ for $m_2\in\{1,..2^{\lfloor nR_2\rfloor}\}$ randomly and independently according to the i.i.d distribution  $\mathcal{Q}_{W_2^n|W_1^n}=\Pi_{t=1}^n\mathcal{Q}_{W_2|W_1}(w_{2t}|w_{1t})$. 

%(how to obtain them? in the first you sum over u^n and w_2^n Q(W^n_1W^n_2U^n)=P_U^nQ(W_2^nW_1^n|U), and for the second you sum over u^n only Q(W^n_1W^n_2U^n)=P_U^nQ(W_2^nW_1^n|U) and you divide by Q_w_1^n, recall that we have a markov chain U-W_1-W_2).
%Sequences $w_2^n(m_2)$ are drawn randomly and independently following the i.i.d. $\Pi_{t=1}^n\mathcal{P}_{W_2}(w_{2t})$, and $2^{\lfloor nR_1 \rfloor}$ sequences $w^n_1(m_1,m_2)$ following $\Pi_{t=1}^n\mathcal{P}(W_{1t}|W_{2t}(m_2))$. 
Since $R_1=\log|\mathcal{U}|=\log|\mathcal{W}_1|$ and $R_2=\log|\mathcal{U}|=\log|\mathcal{W}_2|$,
encoder $\mathcal{E}$ observes $u^n$ and looks in the codebook for the corresponding sequences $w_1^n(m_1)$ and sends $m_1$ to the relay. %and $(m_1,m_2)$ to decoder $\mathcal{D}_{1}$. 
The relay observes $m_1$ and sends $m_2$ to the decoder. Then, the decoder $\mathcal{D}$ observes $m_2$ and declares $v^n$ according to $\tau$.

\subsection{Converse Proof}
Given a triple $(\sigma,\mu,\tau)$ and a random variable $T$ uniformly distributed over $\{1,2,...,n\}$ and independent of $(U^n,M_1,M_2,V^n)$. We identify the auxiliary random variables $W_1 =(M_1,T)$, $W_2 =M_2$, $(U,V)=(U_T,V_{T})$, distributed according to $\mathcal{P}_{UW_1W_2V}^{\sigma\mu\tau}$ defined for all $(u,w_1,w_2,v) = (u_t,x_1,x_2,t,v_{t})$ by
\begin{align*}
  &\mathcal{P}_{UW_1W_2V}^{\sigma\mu\tau} (u,w_1,w_2,v)
  = \mathcal{P}_{U_TW_1W_2TV_{T}}^{\sigma\mu\tau}  (u_t,x_1,x_2,t,v_{t})\\
 =&\frac1n\sum_{u^{t-1}\atop u_{t+1}^n}\sum_{x_1^{t-1},x_{1,t+1}^n \atop x_2^{t-1},x_{2,t+1}^n}\sum_{v^{t-1},v_{t+1}^n}\!\!\!\!\!\!
\bigg(\prod_{t=1}^n\mathcal{P}_{U}(u_t)\bigg)\mathcal{P}^{\sigma}_{M_1|U^n}(m_1|u^n) \nonumber \\ &\times\mathcal{P}^{\mu}_{M_2|M_1}(m_2|m_1)\mathcal{P}^{\tau}_{V^n|M_2}(v^n|m_2). 
\end{align*}
\begin{lemma}\label{lemma:decomposition}
The distribution $\mathcal{P}_{UW_1W_2V}^{\sigma\mu\tau}$ has marginal on $\Delta(\mathcal{U})$ given by $\mathcal{P}_U$ and satisfies the following Markov chain property
\begin{align*}
     U   -\!\!\!\!\minuso\!\!\!\!-  W_1  -\!\!\!\!\minuso\!\!\!\!- W_2, \quad  
     W_1  -\!\!\!\!\minuso\!\!\!\!-  W_2  -\!\!\!\!\minuso\!\!\!\!- V.
 \end{align*}
\end{lemma}
\begin{proof}[Lemma \ref{lemma:decomposition}]
The i.i.d. property of the source ensures that the marginal distribution is $\mathcal{P}_U$. By the definition of the coding functions $\sigma$, $\mu$ and $\tau$ we have 
\begin{align*}
     &(U_T)   -\!\!\!\!\minuso\!\!\!\!-  (M_1,T)  -\!\!\!\!\minuso\!\!\!\!- M_2,\\
     &(M_1,T)  -\!\!\!\!\minuso\!\!\!\!-  M_2  -\!\!\!\!\minuso\!\!\!\!- V_{T}.
 \end{align*}
\end{proof}
Therefore $\mathcal{P}_{UW_1W_2V}^{\sigma\mu\tau} = \mathcal{P}_U\mathcal{P}_{W_1|U}^{\sigma}\mathcal{P}_{W_2|W_1}^{\mu}\mathcal{P}_{V|W_2}^{\tau}$.

\begin{lemma}\label{lemma:singlelettercost}
For all $(\sigma,\tau_1,\tau_2)$ and $i\in\{1,2,3\}$, we have
\begin{align}
   %c_e^n(\sigma,\mu,\tau) =& \E   \big[c_e(U,V)\big] ,\\
   c_i^n(\sigma,\mu,\tau) =& \E
   \big[c_i(U,V)\big] ,
\end{align}
evaluated with respect to $\mathcal{P}_U\mathcal{P}_{W_1|U}^{\sigma}\mathcal{P}_{W_2|W_1}^{\mu}\mathcal{P}_{V|W_2}^{\tau}$. Moreover for all $\sigma,\mu$ we have
\begin{align}
\mathbb{Q}_{3}(\mathcal{P}^{\sigma}_{W_1|U}&,\mathcal{P}^{\mu}_{W_2|W_1}) = \nonumber\\ &\Big\{ \mathcal{Q}_{V|W_2},\;
 \exists \tau \in \mathbb{A}_3(\sigma,\mu),\; \mathcal{Q}_{V|W_2} = \mathcal{P}^{\tau}_{V|W_2}\Big\}, \label{eq:lemmaBRsets}\\
\mathbb{Q}_{2}(\mathcal{P}^{\sigma}_{W_1|U}) &= \Big\{ (\mathcal{Q}_{W_2|W_1},\mathcal{Q}_{V|W_2}),\;
 \exists (\mu,\tau) \in \mathbb{A}_2(\sigma),\; \nonumber \\ &\mathcal{Q}_{W_2|W_1}=\mathcal{P}^{\mu}_{W_1|W_2}, \mathcal{Q}_{V|W_2} = \mathcal{P}^{\tau}_{V|W_2}\Big\}.\label{eq:lemmaBRsets1}
\end{align}
\end{lemma}
\begin{proof}[Lemma \ref{lemma:singlelettercost}] By Definition \ref{def:singlelongrun} we have for  $i\in\{1,2,3\}$
\begin{align}
&c_i^n(\sigma,\mu,\tau) = \nonumber \\&\sum_{u^n, m_1,\atop m_2, v^n}\bigg(\prod_{t=1}^n\mathcal{P}_{U}(u_t)\bigg)
\mathcal{P}^{\sigma}_{M_1|U^n}(m_1|u^n)\mathcal{P}^{\mu}_{M_2|M_1}(m_2|m_1) \nonumber \\
&\times\mathcal{P}^{\tau}_{V^n|M_2}(v^n|m_2)\cdot\Bigg[\frac{1}{n}\sum_{t=1}^n c_i(u_t,v_{t})\Bigg]\\
=& \sum_{t=1}^n \sum_{u_t,x_{1},\atop  x_{2},t,v_{t}} \mathcal{P}^{\sigma,\mu,\tau}(u_t,x_1,x_2,t,v_{t})
\times c_i(u_t,v_{t}) = \E   \big[c_i(U,V)\big] .\nonumber
\end{align}
Given $\mathcal{Q}_{V|W_2}\in \mathbb{Q}_{3}(\mathcal{P}^{\sigma}_{W_1|U},\mathcal{P}^{\mu}_{W_2|W_1})$, we consider $\tau$ such that 
\begin{align*}
    \mathcal{P}^{\tau}_{V^n|M_2}(v^n|m_2) = \prod_{t=1}^n \mathcal{Q}_{V|W_2}(v_{1,t}|m_2).
\end{align*}
Given $(\mathcal{Q}_{W_2|W_1},\mathcal{Q}_{V|W_2})\in \mathbb{Q}_{2}(\mathcal{P}^{\sigma}_{W_1|U})$, we consider $(\mu,\tau)$ such that 
\begin{align*}
\mathcal{P}^{\mu}_{M_2|M_1}(m_2|m_1) =& \prod_{t=1}^n \mathcal{Q}_{W_2|W_1}(m_2|m_1,t),\\
    \mathcal{P}^{\tau}_{V^n|M_2}(v^n|m_2) =& \prod_{t=1}^n \mathcal{Q}_{V|W_2}(v_{1,t}|m_2).
\end{align*}
Therefore
\begin{align}
c^n_{3}(\sigma,\mu,\tau) =& \E_{\mathcal{P}^{\sigma}_{W_1|U}\atop \mathcal{P}^{\mu}_{W_2| W_1}\mathcal{Q}_{V| W_2}}\big[c_3(U,V)\big]\\
=& \min_{\mathcal{P}_{V|W_2}} \E_{\mathcal{P}^{\sigma}_{W_1|U}\atop \mathcal{P}^{\mu}_{W_2| W_1},\mathcal{P}_{V|W_2}}\big[c_3(U,V)\big]\\
\leq& \min_{\tilde{\tau}} \E_{\mathcal{P}^{\sigma}_{W_1|U}\atop \mathcal{P}^{\mu}_{W_2| W_1},\mathcal{P}^{\tilde{\tau}}_{V|W_2}}\big[c_3(U,V)\big]= \min_{\tilde{\tau}} c^n_{3}(\sigma,\mu,\tilde{\tau}),\nonumber
\end{align}
hence $\tau\in \mathbb{A}_3(\sigma,\mu)$.
Similarly, 
\begin{align}
c^n_{2}(\sigma,\mu,\tau) =& \E_{\mathcal{P}^{\sigma}_{W_1|U}\atop \mathcal{Q}_{W_2| W_1}\mathcal{Q}_{V| W_2}}\big[c_2(U,V)\big]\\
=& \min_{(\mathcal{P}_{W_2|W_1},\mathcal{P}_{V|W_2})} \E_{\mathcal{P}^{\sigma}_{W_1|U}\atop \mathcal{P}_{W_2| W_1},\mathcal{P}_{V|W_2}}\big[c_2(U,V)\big]\\
\leq& \min_{(\tilde{\mu},\tilde{\tau})} \E_{\mathcal{P}^{\sigma}_{W_1|U}\atop \mathcal{P}^{\tilde{\mu}}_{W_2| W_1},\mathcal{P}^{\tilde{\tau}}_{V|W_2}}\big[c_2(U,V)\big]= \min_{(\tilde{\mu},\tilde{\tau})} c^n_{2}(\sigma,\mu,\tilde{\tau}),\label{eq:ineqconverse}
\end{align}
and thus $(\mu,\tau)\in \mathbb{A}_2(\sigma)$. The other inclusions are direct and the same arguments imply \eqref{eq:lemmaBRsets1} and \eqref{eq:lemmaBRsets}.
\end{proof}

For any strategy $\sigma$, we have
\begin{align}
\underset{\mu,\tau}{\max} \;  c_1^n(\sigma,\mu,\tau) 
=&\underset{\mu,\tau}{\max} \; \mathbb{E}_{\mathcal{P}^{\sigma}_{W_1|U}\atop\mathcal{P}^{\mu}_{W_{2}|W_1} \mathcal{P}^{\tau}_{V|W_2}}\Big[c_1(U,V)\Big] \label{zachi1}\\ 
\geq&\underset{\mathcal{Q}_{W_2|W_1},\mathcal{Q}_{V|W_2} \atop \in \mathbb{Q}_{2}(\mathcal{Q}_{W_1|U})}{\max}
 \mathbb{E}_{\mathcal{P}^{\sigma}_{W_1|U} \atop \mathcal{Q}_{W_2|W_1} \mathcal{Q}_{V|W_2}}\Big[c_1(U,V)\Big] \label{zachi2}\\ 
\geq&\underset{\mathcal{Q}_{W_1|U}}{\inf}\underset{\mathcal{Q}_{W_2|W_1},\mathcal{Q}_{V|W_2} \atop \in \mathbb{Q}_{2}(\mathcal{Q}_{W_1|U})}{\max}
 {\mathbb{E}}
 \Big[ c_1(U,V) \Big]\label{zachi3}\\ 
=&\Gamma_e(R_1,R_2). \label{optdistoooo}
\end{align} 
Equations \eqref{zachi1} and \eqref{zachi2} comes from Lemma \ref{lemma:singlelettercost}, whereas \eqref{zachi3} comes from taking the infimun over $\mathcal{Q}_{W_1|U}$. This concludes the converse proof of Theorem \ref{2cmdn}.

\section{Locally Restricted Communication} %\vspace{-0.4cm}
\subsection{Relay's Restriction}
Assume that the encoder can send messages at large enough rate $R_1=\log|\mathcal{U}|$, but the relay sends at a fixed smaller rate $R_2$. 
Fix $\mathcal{Q}_{W_1|U}$. In this setting,
%The set $\mathbb{Q}_0(R_2)$ 
 %of target distributions:
    %\begin{align}
%%{%\color{red}{
%\mathbb{Q}_0(R_2) =& \{\mathcal{Q}_{W_2|W_1}; \  \  R_2 \geq I(W_1;W_2) \}. %\subset \Delta(\mathcal{W}_1\times\mathcal{W}_2)^{|\mathcal{U}|},%}} 
  % \label{q10123} %\\
%\mathbb{Q}_1 =& \{\mathcal{Q}_{W_2|W_1} \ s.t. \  R_2 \geq I(W_2;W_1)  \}\subset \Delta(\mathcal{W}_2)^{|\mathcal{W}_1|}.%}} 
 %  \label{q20123}
   % \end{align}
 %The sets %$\mathbb{Q}^1_{1}(\mathcal{Q}_{W_1|U})$,
 %$\mathbb{Q}_{3}(\mathcal{Q}_{W_1|U},\mathcal{Q}_{W_2|U})$  and  $\mathbb{Q}_{2}(\mathcal{Q}_{W_1|U})$ %$\mathbb{Q}_{2}(\mathcal{Q}_{W_2|U})$ 
the single-letter best-responses are defined by:
 \begin{align}%\vspace{-0.4cm}
  &\mathbb{Q}_{3}(\mathcal{Q}_{W_1|U},\mathcal{Q}_{W_2|W_1})= \argmin_{\mathcal{Q}_{V|W_2}}\mathbb{E}%_{\mathcal{P}_U\mathcal{Q}_{W_1|U}\mathcal{Q}_{W_2|W_1}\mathcal{Q}_{V|W_2}}
  [c_3(U,V)], \nonumber %\vspace{-0.4cm}
  \\
 &\mathbb{Q}^r_{2}(\mathcal{Q}_{W_1|U})=\argmin_{(\mathcal{Q}_{W_2|W_1},\mathcal{Q}_{V|W_2})) s.t. R_2\geq I(W_1;W_2), \atop \mathcal{Q}_{V|W_2}\in\mathcal{Q}_3(\mathcal{Q}_{W_1|U},\mathcal{Q}_{W_2|W_1})}\mathbb{E}[c_2(U,V)], %\vspace{-0.4cm} 
 \nonumber
 \end{align} 
  The single-letter optimal cost $\Gamma_e^{r}(R_2)$ of the encoder is given by %\vspace{-0.4cm}
 \begin{align} %\vspace{-0.4cm}
 \Gamma^{r}_e(R_2)=\underset{\mathcal{Q}_{W_1|U}}{\inf}\underset{\mathcal{Q}_{W_2|W_1},\mathcal{Q}_{V|W_2} \atop \in \mathbb{Q}_{2}(\mathcal{Q}_{W_1|U})}{\max}
 {\mathbb{E}}
 \Big[c_1(U,V) \Big]. 
\nonumber \end{align} 
\theorem \label{relrest} Let $R_2\in\R^{+}$. If $R_1=\log|\mathcal{U}|$, then
\begin{align}
\lim_{n\longrightarrow \infty} \Gamma^n_e(R_2)= \underset{n\in\mathbb{N}^{\star}}{\inf} \Gamma^n_e(R_2) = \Gamma^{r}_e(R_2).
%&a) \forall \ \varepsilon>0, \ \exists \hat{n} \in \mathbb{N} \  \forall n \geq \hat{n}, \ \ \Gamma^n_e(R_2) \leq \Gamma_e^{\star}(R_2)  + \varepsilon. \nonumber \\
%&b) \forall n \in \ \mathbb{N},\hspace{2.8cm} \Gamma_e^n(R_2) \geq \Gamma_e^{\star}(R_2) . \nonumber
%\end{align}
%Using Fekete's Lemma for the sub-additive sequence $\big(n \Gamma_e^n(R_1,R_2) \big)_{n\in \N^{\star}}$ %\cite[Lemma 1]{rouphael2021strategic}
%we get
%\begin{align}
%\lim_{n\longrightarrow \infty} \Gamma^n_e(R_1,R_2)= \underset{n\in\mathbb{N}^{\star}}{\inf} \Gamma^n_e(R_1,R_2) = \Gamma^{\star}_e(R_1,R_2).
\end{align}

The proof of Theorem \ref{relrest} relies on the lossy source coding at the relay by considering the source to be the observed message which is uniformly drawn from the codebook of size $2^{nR_1}$. This slight modification does not affect the condition on the covering lemma as the coding will only depend on the size $2^{nR_2}$ of the message set of the relay.

%\vspace{-2cm}
\subsection{Encoder's Restriction} %\vspace{-0.4cm}
Now assume that $R_2=\log|\mathcal{U}|$, i.e. the encoder is restricted to a limited amount of bits per transmission, but the relay can transmit with no information constraints. Therefore,  
the set $\mathbb{Q}^e_0(R_1)$ 
 of the encoder's target distributions is given by
    \begin{align} \nonumber%\vspace{-0.4cm}
%%{%\color{red}{
\mathbb{Q}^e_0(R_1) =& \{\mathcal{Q}_{W_1|U}; \  \  R_1 \geq I(U;W_1) \}. %\subset \Delta(\mathcal{W}_1\times\mathcal{W}_2)^{|\mathcal{U}|},%}} 
   %\\
%\mathbb{Q}_1 =& \{\mathcal{Q}_{W_2|W_1} \ s.t. \  R_2 \geq I(W_2;W_1)  \}\subset \Delta(\mathcal{W}_2)^{|\mathcal{W}_1|}.%}} 
 %  \label{q20123}
    \end{align}
Single-letter best-responses are defined by:
 \begin{align} \nonumber
  &\mathbb{Q}_{3}(\mathcal{Q}_{W_1|U},\mathcal{Q}_{W_2|W_1})= \argmin_{\mathcal{Q}_{V|W_2}}\mathbb{E}%_{\mathcal{P}_U\mathcal{Q}_{W_1|U}\mathcal{Q}_{W_2|W_1}\mathcal{Q}_{V|W_2}}
  [c_3(U,V)], %\vspace{-0.4cm} 
  \\
 &\mathbb{Q}^e_{2}(\mathcal{Q}_{W_1|U})=\argmin_{(\mathcal{Q}_{W_2|W_1},\mathcal{Q}_{V|W_2})), \atop \mathcal{Q}_{V|W_2}\in\mathcal{Q}_3(\mathcal{Q}_{W_1|U},\mathcal{Q}_{W_2|W_1})}\mathbb{E}[c_2(U,V)], %\vspace{-0.4cm} 
 \nonumber
 \end{align}
The single-letter optimal cost $\Gamma_e^{\star}(R_1)$ of the encoder is given by%\vspace{-0.4cm}
 \begin{align} %\vspace{-0.4cm}
 \Tilde{\Gamma}_e(R_1)=\underset{\mathcal{Q}_{W_1|U}\in\mathcal{Q}^e_0(R_1)}{\inf}\underset{\mathcal{Q}_{W_2|W_1},\mathcal{Q}_{V|W_2} \atop \in \mathbb{Q}^e_{2}(\mathcal{Q}_{W_1|U})}{\max} %\vspace{-0.4cm}
 {\mathbb{E}}
 \Big[c_1(U,V) \Big]. 
\nonumber\end{align} 
\nonumber
\theorem Let $R_1\in\R^{+}$. If $R_2=\log|\mathcal{U}|$, then %\vspace{-0.4cm}
\begin{align} %\vspace{-0.4cm}
\lim_{n\longrightarrow \infty} \Gamma^n_e(R_2)= \underset{n\in\mathbb{N}^{\star}}{\inf} \Gamma^n_e(R_2) = \Tilde{\Gamma}_e(R_2).
\end{align}

%\vspace{-2cm}
\section{Locally Cooperating Agents} Consider now that either the relay and the encoder or the relay and the decoder are cooperating. In other words, we assume that either $c_1=c_2$ or $c_2=c_3$ holds. 
\subsection{Encoder-Relay Cooperation}
Assume $c_1=c_2$ the encoder and the relay are cooperating. The encoder will reveal information using the maximal rate $R_1$. 
The set $\mathbb{Q}^s_0(R_1,R_2)$ 
 of target distributions:
    \begin{align} %\vspace{-0.4cm}
\mathbb{Q}^s_0(R_1,R_2) =& \{\mathcal{Q}_{W_1|U}; \  \  R_1 \geq I(U;W_1) \}. 
    \end{align}
Single-letter best-responses are defined by:
 \begin{align}
  &\mathbb{Q}_{3}(\mathcal{Q}_{W_1|U},\mathcal{Q}_{W_2|W_1})= \argmin_{\mathcal{Q}_{V|W_2}}\mathbb{E}%_{\mathcal{P}_U\mathcal{Q}_{W_1|U}\mathcal{Q}_{W_2|W_1}\mathcal{Q}_{V|W_2}}
  [c_3(U,V)], \\
 &\mathbb{Q}^s_{2}(\mathcal{Q}_{W_1|U})=\argmin_{(\mathcal{Q}_{W_2|W_1},\mathcal{Q}_{V|W_2})), R_2\geq I(W_1;W_2) \atop \mathcal{Q}_{V|W_2}\in\mathcal{Q}_3(\mathcal{Q}_{W_1|U},\mathcal{Q}_{W_2|W_1})}\mathbb{E}[c_2(U,V)], %\vspace{-0.4cm}
 \end{align} %\vspace{-0.4cm}
  The single-letter optimal cost $\Gamma_e^{s}(R_1,R_2)$ of the encoder is given by
 \begin{align} %\vspace{-0.4cm}
 \Gamma^{s}_e(R_1,R_2)=\underset{\mathcal{Q}_{W_1|U}\in\mathcal{Q}^s_0(R_1,R_2)}{\inf}\underset{\mathcal{Q}_{W_2|W_1},\mathcal{Q}_{V|W_2} \atop \in \mathbb{Q}^s_{2}(\mathcal{Q}_{W_1|U})}{\max}
 {\mathbb{E}}
 \Big[c_1(U,V) \Big]. \nonumber
\end{align} 
%\vspace{-0.4cm}
\theorem Let $(R_1,R_2)\in\R^2_{+}$. If $c_1=c_2$, then 
\begin{align} %\vspace{-0.6cm}
\lim_{n\longrightarrow \infty} \Gamma^n_e(R_1,R_2)= \underset{n\in\mathbb{N}^{\star}}{\inf} \Gamma^n_e(R_1,R_2) = \Gamma^{s}_e(R_1,R_2).
\end{align}
The proof relies on considering that the relay observes the source as the encoder can fully reveal it, and the Bayesian persuasion setting between the relay and the decoder.

%\vspace{-2cm}
\subsection{Relay-Decoder Cooperation}
Assume now that the decoder cooperates with the relay because $c_2=c_3$.
The set $\mathbb{Q}^d_0(R_1,R_2)$ 
 of target distributions:
    \begin{align} %\vspace{-0.4cm}
\mathbb{Q}^d_0(R_1,R_2) =& \{\mathcal{Q}_{W_1|U}; \  \  R_1 \geq I(U;W_1) \}. 
    \end{align}
Single-letter best-responses are defined by:
 \begin{align} %\vspace{-0.4cm}
  &\mathbb{Q}_{3}(\mathcal{Q}_{W_1|U},\mathcal{Q}_{W_2|W_1})= \argmin_{\mathcal{Q}_{V|W_2}}\mathbb{E}%_{\mathcal{P}_U\mathcal{Q}_{W_1|U}\mathcal{Q}_{W_2|W_1}\mathcal{Q}_{V|W_2}}
  [c_3(U,V)], \\
 &\mathbb{Q}^d_{2}(\mathcal{Q}_{W_1|U})=\argmin_{(\mathcal{Q}_{W_2|W_1},\mathcal{Q}_{V|W_2})), R_2\geq I(W_1;W_2) \atop \mathcal{Q}_{V|W_2}\in\mathcal{Q}_3(\mathcal{Q}_{W_1|U},\mathcal{Q}_{W_2|W_1})}\mathbb{E}[c_2(U,V)].%\vspace{-0.4cm}
 \end{align} %\vspace{-0.4cm}
  The single-letter optimal cost $\Gamma_e^{d}(R_1,R_2)$ of the encoder is given by
 \begin{align} %\vspace{-0.4cm}
 \Gamma^{d}_e(R_1,R_2)=\underset{\mathcal{Q}_{W_1|U}\in\mathcal{Q}^d_0(R_1,R_2)}{\inf}\underset{\mathcal{Q}_{W_2|W_1},\mathcal{Q}_{V|W_2} \atop \in \mathbb{Q}^d_{2}(\mathcal{Q}_{W_1|U})}{\max}
 {\mathbb{E}}
 \Big[c_1(U,V) \Big]. \nonumber %\vspace{-0.4cm}
\end{align} 
%\vspace{-0.4cm}
\theorem
\begin{align} %\vspace{-0.4cm}
\lim_{n\longrightarrow \infty} \Gamma^n_e(R_1,R_2)= \underset{n\in\mathbb{N}^{\star}}{\inf} \Gamma^n_e(R_1,R_2) = \Gamma^{d}_e(R_1,R_2).
\end{align}
The proof follows by considering the relay and the decoder as one party, and lossy source coding at the encoder.

\section{Binary Example}
Assume that $R_1=R_2=\log|\mathcal{U}|$, and $c_1=c_3$. We illustrate the problem using a binary source information $\mathcal{U}= \{u_0,u_1\}$, binary channel inputs $\mathcal{X}_2=\{x^2_{0},x^2_{1}\}$,  $\mathcal{X}_1=\{x^1_{0},x^1_{1}\}$ and binary action set $\mathcal{V}=\{v_0,v_1\}$. The prior belief $\mathcal{P}_U(u_1)$ is given by the parameter $p_0\in [0,1]$. Single-letter cost functions are given in the tables below.
  \begin{table}[!htb]
   % \caption{Global caption}
%\end{table}
%\end{minipage}%
%\begin{minipage}{.5\linewidth}
\begin{minipage}{.2\linewidth}
%\begin{table}
\caption{$c_1(u,v)$}
\centering 
\begin{tabular}{|l|l|l|}
\hline
$\ $ & $v_0$ & $v_1$ \\
\hline
$u_0$ &$9$ & $0$ \\
\hline
$u_1$ & $4$ & $10$\\
\hline
\end{tabular}
\end{minipage}
\qquad 
\quad 
\begin{minipage}{.2\linewidth}
%\begin{table}
 \caption{ $c_2(u,v)$}
      \centering \begin{tabular}{|l|l|l|}
     
\hline
$\ $ & $v_0$ & $v_1$ \\
\hline
$u_0$ &$1$ & $0$ \\
\hline
$u_1$ & $1$ & $0$\\
\hline 
\end{tabular}
\end{minipage}
\qquad 
\quad 
\begin{minipage}{.2\linewidth}
%\begin{table}
 \caption{$c_3(u,v)$}
      \centering \begin{tabular}{|l|l|l|}
     
\hline
$\ $ & $v_0$ & $v_1$ \\
\hline
$u_0$ &$9$ & $0$ \\
\hline
$u_1$ & $4$ & $10$\\
\hline 
\end{tabular}
\end{minipage}
\end{table}

Let $(\alpha,\beta), (\gamma,\delta),$and $ (\epsilon,\eta) \in [0,1]^2$ %such that $\alpha+\beta \neq 1$
.
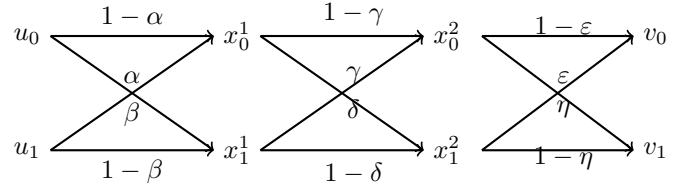
\begin{figure}[ht!]
\begin{center}
\begin{tikzpicture}[xscale=3.1,yscale=3.8, thick, empty dot/.style = { circle, draw, fill = white!0,
                           inner sep = 0pt, minimum size = 2.8pt },
     filled dot/.style = { empty dot, fill = black},
      green dot/.style = { empty dot, fill = green},
     red dot/.style = { empty dot, fill = red}
   ]
   \def\r{3}
%\draw[thick,->](-0.1,1)--(1.1,1);
\draw[thick,->](0.4,0.7)--(1.1,0.3);
\draw[thick,->](0.4,0.7)--(1.1,0.7);
%\draw[thick,->](-0.1,1)--(1.1,0);
%\draw[thick,->](-0.1,0)--(1.1,0);
\draw[thick,->](0.4,0.3)--(1.1,0.7);
\draw[thick,->](0.4,0.3)--(1.1,0.3);
%\draw[thick,->](-0.1,0)--(1.1,1);
\node [above, black] at (0.75,0.5) {$\alpha$};
%\node [above, black] at (0.5,1) {$\alpha_1$};
\node [above, black] at (0.75,0.7) {$1-\alpha$};
%\node [above, black] at (0.5,0.5) {$\alpha_4$};
\node [below, black] at (0.75,0.5) {$\beta$};
\node [below, black] at (0.75,0.3) {$1-\beta$};
%\node [below, black] at (0.5,0) {$\beta_4$};
%\node [above, black] at (0.5,0.35) {$\beta_1$};
\node [left, black] at (0.4,0.7) {$u_{0}$};
\node [left, black] at (0.4,0.3) {$u_{1}$};
%\node [right, black] at (1.1,0) {$x_{11}$};
\node [right, black] at (1.1,0.7) {$x^1_{0}$};
\node [right, black] at (1.1,0.3) {$x^1_{1}$};
%\node [right, black] at (1.1,1) {$x_{00}$};
%\node [right, black] at (2,0) {$(y_{1,1},y_{2,1})$};
\node [right, black] at (2,0.7) {$x^2_{0}$};
\node [right, black] at (2,0.3) {$x^2_{1}$};
%\node [right, black] at (2,1) {$(y_{1,0},y_{2,0})$};
\draw[thick,->](2.25,0.3)--(2.9,0.7);
\draw[thick,->](2.25,0.7)--(2.9,0.3);
\draw[thick,->](2.25,0.3)--(2.9,0.3);
\draw[thick,->](2.25,0.7)--(2.9,0.7);
%\node [right, black] at (3,0) {$(q_{1,1},q_{2,1})$};
\node [right, black] at (2.9,0.7) {$v_{0}$};
\node [right, black] at (2.9,0.3) {$v_{1}$};
%\node [right, black] at (3,1) {$(q_{1,0},q_{2,0})$};
%\draw[thick,->](1.3,1)--(2,0.7);
%\draw[thick,->](1.3,0)--(2,0.3);
\draw[thick,->](1.3,0.3)--(2,0.3);
\draw[thick,->](1.3,0.7)--(2,0.7);
\draw[thick,->](1.3,0.3)--(2,0.7);
\draw[thick,->](1.3,0.7)--(2,0.3);
\node [above, black] at (1.7,0.38) {$\delta$};
%\node [above, black] at (1.7,0.82) {$1$};
\node [above, black] at (1.7,0.5) {$\gamma$};
\node [above, black] at (1.7,0.15) {$1-\delta$};
\node [below, black] at (1.7,0.85) {$1-\gamma$};
%\node [below, black] at (1.7,0.16) {$1$};
\node [above, black] at (2.6,0.5) {$\varepsilon$};
\node [above, black] at (2.6,0.2) {$1-\eta$};
\node [below, black] at (2.6,0.8) {$1-\varepsilon$};
\node [above, black] at (2.6,0.38) {$\eta$};
\end{tikzpicture}
\caption {Encoders' Joint Strategies $\sigma_1$ and $\sigma_2$ and decoder's strategy $\sigma_3$. %and $\varepsilon_2$ 
}
\label{jostrachaIII}
\end{center}
\end{figure}

Using Baye's rule, we compute the following 
\begin{align}
    &q^1_{0}=\mathcal{P}(u_1|x^2_{0})=\frac{\mathcal{P}(u_1,x^2_{0})}{\mathcal{P}(x^2_{0})}= \nonumber \\ &\frac{(\beta(1-\gamma)+(1-\beta)\delta)\cdot p_0}{(\beta(1-\gamma)+(1-\beta)\delta)\cdot p_0 + ((1-\alpha)(1-\gamma)+\alpha\delta)\cdot (1-p_0)}, \label{qqqqq10123}
   % &q^1_{1}=\mathcal{P}(u_1|x^2_{1})=\frac{\mathcal{P}(u_1,x^2_{1})}{\mathcal{P}(x^2_{1})},
    %=\frac{(\beta\gamma+(1-\beta)(1-\delta))\cdot p_0}{(\beta\gamma+(1-\beta)(1-\delta)) \cdot p_0 +((1-\alpha)\gamma+\alpha(1-\delta))\cdot(1-p_0)}
    %\label{qqqqqqqq11123}\\
     %  &q^2_{0}=\mathcal{P}(x^1_1|x^2_{0})=\frac{\mathcal{P}(x^1_1,x^2_{0})}{\mathcal{P}(x^2_{0})},%=\frac{(\alpha\delta)\cdot(1-p_0)+(1-\beta)\delta\cdot p_0}{(\beta(1-\gamma)+(1-\beta)\delta)\cdot p_0+ (\alpha\delta+(1-\alpha)(1-\gamma))\cdot (1-p_0)}
      % \label{qqqqq210123}\\
    %&q^2_{1}=\mathcal{P}(x^1_1|x^2_{1})=\frac{\mathcal{P}(x^1_1,x^2_{1})}{\mathcal{P}(x^2_{1})}.
   % =\frac{\alpha(1-\delta)\cdot(1-p_0)+(1-\beta)(1-\delta)\cdot p_0}{(\beta\gamma+(1-\beta)(1-\delta)) \cdot p_0 +((1-\alpha)\gamma+\alpha(1-\delta))\cdot(1-p_0)}.\label{qqqqqqqq211123}
    \end{align}
    Similarly, one can compute $q^1_{1}=\mathcal{P}(u_1|x^2_{1}),
 q^2_{0}=\mathcal{P}(x^1_1|x^2_{0})$ and 
 $q^2_{1}=\mathcal{P}(x^1_1|x^2_{1})$ can be computed.
%Let $p_1=\alpha(1-p_0)+(1-\beta)p_0$. Thus, $\mathcal{P}(x_0^1)=1-p_1$ and $\mathcal{P}(x_1^1)=p_1$. 
 Let $p_1(\alpha,\beta)\in[0,1]$ denote the belief parameter of the decoder about $X^1$. In other words, $p_1(\alpha,\beta)=\mathcal{P}_{X^1}(x^1_1)= \sum_{u}\mathcal{P}_{U}(u)\cdot\mathcal{P}_{X^1|U}(x_1^1|u)=(1-p_0)\alpha + p_0(1-\beta)$. Beliefs  $q^2_{0}$ and $q^2_{1}$ can be reformulated as follows:
\begin{align}
    q^2_{0}=&\frac{p_1(\alpha,\beta)\cdot\delta}{p_1(\alpha,\beta)\cdot\delta+(1-p_1(\alpha,\beta))(1-\gamma)}, \\
    q^2_{1}=&\frac{p_1(\alpha,\beta)\cdot(1-\delta)}{p_1(\alpha,\beta)\cdot(1-\delta)+(1-p_1(\alpha,\beta))\gamma}.
\end{align}
Knowing $(\alpha,\beta)$, the relay will tune $(\gamma^{\star}(\alpha,\beta),\delta^{\star}(\alpha,\beta))$ so that posteriors $(q^2_0,q^2_1)$ are an optimal splitting as follows:
\begin{align}
    \gamma^{\star}(\alpha,\beta) =& \frac{(1-q_1^2)(p_1(\alpha,\beta)-q_0^2)}{(1-p_1(\alpha,\beta))(q_1^2-q_0^2)}, \\
    \delta^{\star}(\alpha,\beta) =& \frac{q_0^2(q_1^2-p_1(\alpha,\beta))}{p_1(\alpha,\beta)\cdot(q_1^2-q_0^2)}.
\end{align}

%Therefore, encoder $1$ will plug in the respective values of $(\gamma^{\star}(\alpha,\beta),\delta^{\star}(\alpha,\beta))$ in \eqref{qqqqq10123} and \eqref{qqqqqqqq11123} and solve for $(\alpha,\beta)$ such that the splitting $(q^1_0,q^1_1)$ is optimal. 
\begin{remark}
The order of commitment is crucial in this setting. If the relay commits to a strategy $(\gamma,\delta)$ and announces it before the encoder commits to and announces a strategy, thus %the encoder's optimal strategy would be as follows, where $(q_0^1,q_1^1)$ are an optimal splitting:  
%\begin{align}
 %   \beta^{\star}(\gamma,\delta)&=\frac{q_0^1[(1-\delta)(1-q_1^1)p_0-\delta q_1^1(1-p_0)(2\gamma-1)]+q_1^1(1-2\delta)[(1-\gamma)q_0^1(1-p_0)-p_0\gamma]}{p_0[q_1^1((1-2\delta)^2(1-q_0^1))+q_0^1((2\delta-1)(1-q_1^1))]},\\
  %  \alpha^{\star}(\gamma,\delta)&=\frac{(1-\gamma)(1-p_0)(2\delta-1)(1-q_1^1)p_0+(1-\delta)(1-q_1^1)p_0-\delta q_1^1(1-p_0)}{q_1^1(1-p_0)(1-2\delta)-(1-q_1^1)(2\gamma-1)(1-p_0)(2\delta-1)p_0}.
%\end{align}
%Note that in this case, 
the problem boils down to the one tackled in \cite{jet}.
\end{remark}
\begin{remark}
If $\alpha+\beta=1$, then the source $U$ and channel's input $X^1$ are independent. In that case, the decoder will stick to its prior belief $p_0$ disregarding any information received from the relay, and play its default action $v_1$. The corresponding costs are $10\times0.4=4$ for the encoder, and $1$ for the relay.
\end{remark}

%But what is the optimal splitting $(q^2_0,q^2_1)$ for the relay? \\ 
%To answer this question, we consider the subgame given in Fig. \ref{1subgamed1d3}. 
Using the convex closure of the decoder's expected cost, we aim to find the optimal splitting $(q^2_0,q^2_1)$ for the relay. For that, we need to define the costs $ c^x_i(x^1,v), i\in\{1,2,3\}$ of all players as functions of the channel input $X_1$ and the decoder's action $V$ as follows\begin{align}
    c^x_1(x^1,v)=\sum_u\mathcal{P}_{U|X^1}(u|x^1)c_1(u,v), \ \forall x^1,v.\\
    c^x_2(x^1,v)=\sum_u\mathcal{P}_{U|X^1}(u|x^1)c_2(u,v), \ \forall x^1,v.\\
    c^x_3(x^1,v)=\sum_u\mathcal{P}_{U|X^1}(u|x^1)c_3(u,v). \ \forall x^1,v.
\end{align}
The distributions $\mathcal{P}_{U|X^1}(u_0|x^1_{0})$ and $\mathcal{P}_{U|X^1}(u_1|x^1_{1})$ computed as follows
\begin{align}
  \mathcal{P}(u_0|x^1_{0})=&\frac{\mathcal{P}(u_0,x^1_{0})}{\mathcal{P}(x^1_{0})}=\frac{(1-\alpha)\cdot(1-p_0)}{\beta\cdot p_0 + (1-\alpha)\cdot (1-p_0)}, %=\frac{(1-\alpha)\cdot(1-p_0)}{1-p_1(\alpha,\beta)}, 
  \label{qqqqqx10123}\\
 \mathcal{P}(u_1|x^1_{1})=&\frac{\mathcal{P}(u_1,x^1_{1})}{\mathcal{P}(x^1_{1})}%\frac{(1-\mathcal{P}^{\sigma}(y_{1,0}|u_1))\cdot(p_0)}{(1-\mathcal{P}^{\sigma}(y_{1,0}|u_1))\cdot(p_0)+\mathcal{P}^{\sigma}(y_{1,1}|u_0)\cdot(1-p_0)} \nonumber \\ 
    =\frac{(1-\beta)\cdot p_0}{(1-\beta) \cdot p_0 +\alpha\cdot(1-p_0)}%=\frac{(1-\beta)\cdot p_0}{p_1(\alpha,\beta)},
    .\label{qqqqqqqqx11123}
    \end{align}

\begin{figure}[ht!]
\centering
\begin{tikzpicture}[xscale=0.3,yscale=0.3]
\draw[thick,->](-0.5,0)--(11,0);
\draw[thick,->](0,-0.5)--(0,11);
\draw(0,9)--(10,4);
\draw(0,0)--(10,10);
\draw[thick](10,0)--(10,11);
\draw[dotted](4,0)--(4,11);
\node [above, black] at (9,10) {$v_1$};
\node [above, black] at (9,5) {$v_0$};
\draw[dotted](6,0)--(6,6);
\draw [ultra thick][red] (6, 6) -- (0,0);
\draw [ultra thick][red] (6, 6) -- (10,4);
\node [below, black] at (6,0) {$g$};
\node [below, black] at (10,0) {1};
\node [left, black] at (0,9) {9};
\node [below, black] at (4,0) {$p_0$};
\node [right, black] at (10,4) {4};
\node [right, black] at (10,10) {10};
\draw[thick][dotted][green](0,0)--(10,4);
\draw[dotted](0,1.6)--(4,1.6);
\node [left, black] at (0,1.6) {$\Gamma_e%=0.667
$};
\node [below, black] at (12,0) {$\mathrm{P}(u_1)$};
\end{tikzpicture}
\begin{tikzpicture}[xscale=0.3,yscale=0.3]
\draw[thick,->](-0.5,0)--(11,0);
\draw[thick,->](0,-0.5)--(0,11);
\draw(0,0)--(10,0);
\draw[thick](10,0)--(10,11);
\draw[dotted](6,0)--(6,11);
\draw[dotted](0,3.35)--(4,3.35);
\node [above, black] at (9,10) {$v_0$};
\node [above, black] at (5,0.1) {$v_1$};
\draw[dotted](4,0)--(4,11);
\draw [ultra thick][red] (6, 0) -- (10,0);
\draw [ultra thick][red] (0, 10) -- (6,10);
\node [below, black] at (6.5,0) {$g$};
%%\node [below, black] at (10,0) {1};
\node [below, black] at (4,0) {$p_0$};
\node [right, black] at (10,10) {1};
%\node [below, black] at (0,0) {$q^1_1$};
\node [left, black] at (0,3.35) {$C_2^{\star}%=0.667
$};
\draw[thick][dotted][green](6,0)--(0,10);
\node [below, black] at (12,0) {$\mathrm{P}(u_1)$};
\end{tikzpicture}
\begin{tikzpicture}[xscale=0.3,yscale=0.3]
\draw[thick,->](-0.5,0)--(11,0);
\draw[thick,->](0,-0.5)--(0,11);
\draw(0,9)--(10,4);
\draw(0,0)--(10,10);
\draw[thick](10,0)--(10,11);
\draw[dotted](4,0)--(4,11);
\node [above, black] at (9,10) {$v_1$};
\node [above, black] at (9,5) {$v_0$};
\draw[dotted](6,0)--(6,6);
\draw [ultra thick][red] (6, 6) -- (0,0);
\draw [ultra thick][red] (6, 6) -- (10,4);
\node [below, black] at (6,0) {$g$};
\node [below, black] at (10,0) {1};
\node [left, black] at (0,9) {9};
\node [below, black] at (4,0) {$p_0$};
\node [right, black] at (10,4) {4};
\node [right, black] at (10,10) {10};
\node [below, black] at (12,0) {$\mathrm{P}(u_1)$};
%\draw[thick,->](-0.5,0)--(11,0);
%\draw[thick,->](0,-0.5)--(0,11);
%\draw(0,9)--(10,4);
%\draw(0,0)--(10,10);
%\draw[thick](10,0)--(10,11);
%\draw[dotted](8,0)--(8,11);
%\node [above, black] at (9,10) {$v_1$};
%\node [above, black] at (9,5) {$v_0$};
%\draw[dotted](6,0)--(6,6);
%\draw [thick][red] (6, 6) -- (0,0);
%\draw [thick][red] (6, 6) -- (10,4);
%\node [below, black] at (6,0) {$\gamma$};
%\node [below, black] at (10,0) {1};
%\node [left, black] at (0,9) {9};
%\node [below, black] at (8,0) {$p_0$};
%\node [right, black] at (10,4) {4};
%\node [right, black] at (10,10) {10};
%\node [below, black] at (12,0) {$\mathrm{P}(u_1)$};
%%\node [below, black] at (5,-1.5) {Receiver's Expected Utility with $p_0=0.4, \ \gamma =0.6$};
\end{tikzpicture}
%%\subcaption 
\\
\hspace{1cm}\caption{Expected cost functions with $p_0=0.4$, $g=0.6$, $C_2^{\star}=0.33$ and $\Gamma_e=1.6$ for large enough rates $R_1, R_2 \geq \log|\mathcal{U}|=1$.}
\label{fig:encoderexpected123}
\end{figure}
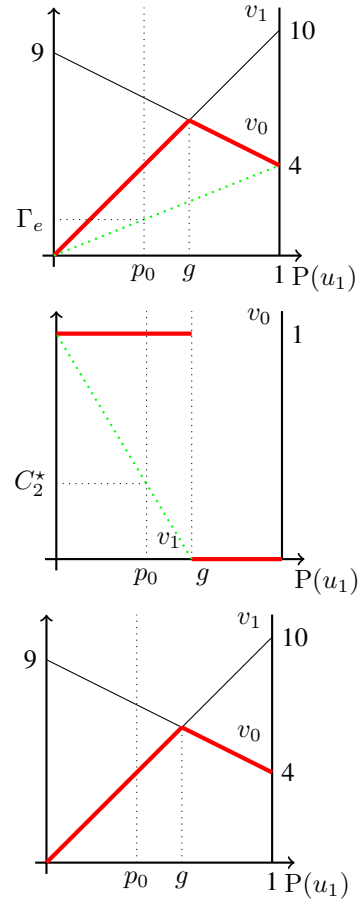

    \begin{figure}
        \centering
    \includegraphics[width=10cm]{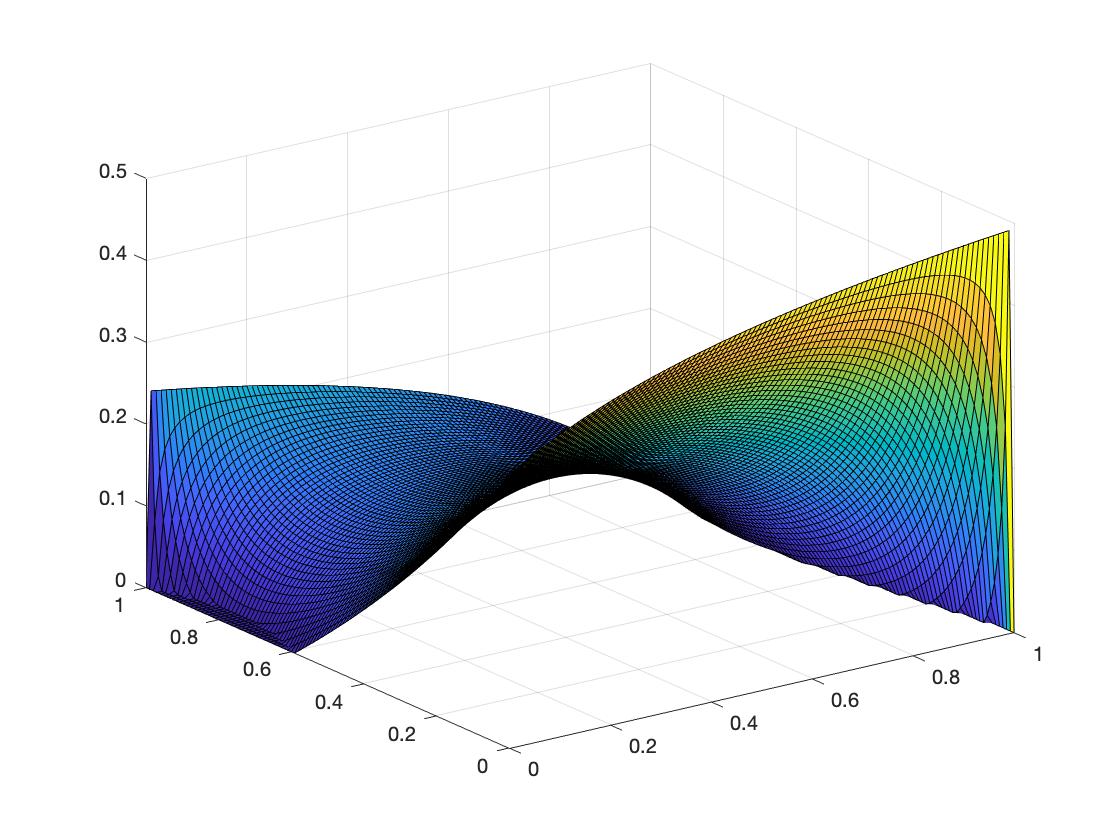}
   \caption{\textit{The relay's optimal cost $C_2^{\star}(\alpha,\beta)$ for $p_0=0.4$ and $\alpha,\beta \in [0,1]$.}}
    \label{fig:my_label}
   \end{figure}

%For a given $(\alpha,\beta)$, the prior belief of the decoder in the subgame is denoted by the parameter $p_1(\alpha,\beta) \in [0,1]$ such that $p_1(\alpha,\beta)=\mathcal{P}_{X^1}(x^1_1)= \sum_{u}\mathcal{P}_{U}(u)\cdot\mathcal{P}_{X^1|U}(x_1^1|u)=(1-p_0)\alpha + p_0(1-\beta)$. \\
%Single-letter costs $c^x_2(x^1,v)$ and $c^x_3(x^1,v)$ of the relay and decoder $3$ can be reformulated as functions of the channel input $W_1$ and decoder's action $V$ as follows
%\begin{align}
 %   c^x_i(x^1,v)=\sum_{u}\mathcal{P}_{U|W_1}(u|x^1)c_i(u,v), \qquad i \in \{2,3\}.
%\end{align}
%By setting $(1-g) \cdot (9\mathcal{P}(u_0|x^1_0)+4(1-\mathcal{P}(u_0|x^1_0))) + g\cdot( 4\mathcal{P}(u_1|x^1_1)+(1-\mathcal{P}(u_1|x^1_1))9)= g \cdot 10\mathcal{P}(u_1|x^1_1) +(1-g)\cdot(1-\mathcal{P}(u_0|x^1_0)10)$, we get 
The threshold $g(\alpha,\beta)$ at which the decoder changes action is computed as follows \begin{align}
    g(\alpha,\beta)&=\frac{2-\mathcal{P}(u_0|x^1_0)\cdot 5}{5\cdot(1-\mathcal{P}(u_1|x^1_1)-\mathcal{P}(u_0|x^1_0))}% \\ &= 
    %\frac{\frac{(1-\alpha)\cdot(1-p_0)}{\beta\cdot p_0 + (1-\alpha)\cdot (1-p_0)}\cdot 3}{2 \frac{(1-\beta)\cdot p_0}{(1-\beta) \cdot p_0 +\alpha\cdot(1-p_0)}+3\frac{(1-\alpha)\cdot(1-p_0)}{\beta\cdot p_0 + (1-\alpha)\cdot (1-p_0)}} 
  %  \\ &= 
   % \frac{(1-\alpha)\cdot(1-p_0)\cdot((1-\beta) \cdot p_0 +\alpha\cdot(1-p_0))\cdot 3}{2\cdot(1-\beta)\cdot p_0\cdot(\beta\cdot p_0 + (1-\alpha)\cdot (1-p_0))+(1-\alpha)\cdot(1-p_0)\cdot((1-\beta) \cdot p_0 +\alpha\cdot(1-p_0))\cdot3}
\end{align}
We define the single-letter cost of the encoder and the relay as a function of the belief parameter $q \in [0,1]$ about $X^1$ and for threshold $g$ as follows:
\begin{align}
  c^x_1(q)=&\sum_{x^1}q(x^1)c_1^x(x^1,v^{\star}(q(x^1)))%=\mathbbm{1}_{\{q<g(\alpha,\beta)\}}\cdot [q(x^1_1)\cdot\mathcal{P}(u_1|x_1^1)\cdot10] +\nonumber\\ &\mathbbm{1}_{\{q\geq g(\alpha,\beta)\}}\cdot [q(x_0^1)\cdot\mathcal{P}(u_0|x_0^1)\cdot 9+q(x_1^1)\cdot\mathcal{P}(u_1|x_1^1)\cdot 4]
  , \\  
   c^x_2(q)=&\sum_{x^1}q(x^1)c_2^x(x^1,v^{\star}(q(x^1))) %= \mathbbm{1}_{\{q<g(\alpha,\beta)\}}\cdot [ q(x_0^1)\cdot\mathcal{P}(u_0|x_0^1)+q(x_1^1)\cdot\mathcal{P}(u_1|x_1^1)]
   ,%  c_2(q)=\mathbbm{1}_{\{q<g\}}.
\end{align}
where \begin{align}
    v^{\star}(q(x^1))=\argmin_{v}\sum_{x^1}q(x^1)c_3^x(x^1,v).
\end{align}
%We can derive a similar definition for the single-letter cost $c_2^x(p)$ where $p \in [0,1]$ denotes the belief parameter about $W_1$,
%\begin{align}
 %   c^x_2(p)=\sum_{u}\mathcal{P}_{U|W_1}(u|x^1)\mathbbm{1}_{\{p<g\}}.
%\end{align}
%Hence, the cost $C_1(0.3)= \lambda^1_0c_1(q^1_0)+\lambda_1^1c_1(q^1_1) = 0$.

For a given $(\alpha,\beta)$, %\lambda_{x^2} \in \Delta(\mathcal{X}^2 and q_{x^2}\in\Delta(\mathcal{X}^1). 
 the optimal cost of the relay can be computed using the convexification method as follows:
\begin{align}
    C_2^{\star}(\alpha,\beta)=&\inf_{(\lambda,q)_{x^2}}\Big\{\sum_{x^2}\lambda_{x^2}c^x_2(q_{x^2}), \  \nonumber \\ &\sum_{k}\lambda^2_k =1,  \ \sum_{x^2}\lambda_{x^2}q_{x^2} =p_1(\alpha,\beta) \Big\}.
\end{align}
%We denote the set of optimal splittings $(\lambda^2,p^2)_n$ for a pair $(\alpha,\beta)$ by $\mathcal{L}(\alpha,\beta)$.
%Define the single-letter cost of the encoder $c_1(q)$ for a belief parameter $q\in[0,1]$ as follows
%\begin{align}
 %  c_1(q)= \begin{cases}
  %  10q, \quad &\textrm{if } q\leq \gamma \\
   % 9(1-q)+4q, \quad &\textrm{if } q >\gamma
%    \end{cases}
%\end{align}

%For a given $(\alpha,\beta)$, the cost that could be achieved by the encoder is derived as follows:
%\begin{align}
 %   C_1(\alpha,\beta)=%\sum_{w}\lambda_w^1c_1(q_w^1), \  \quad \sum_{w}\lambda_w^1=1,  \ \sum_{w}\lambda_w^1 q_w^1 =p_0. %\inf_{(\lambda^2,p^2)_w \in \mathcal{L}(\alpha,\beta)}\{\sum_{w}\lambda_w^1(\lambda_w^2)c_1(q_w^1(p_w^2)), \  \quad \sum_{w}\lambda_w^1(\lambda^2_w) =1,  \ \sum_{w}\lambda_w^1(\lambda_w^2)q_w^1(p_w^2) =p_0 
    %\}.
%\end{align}

The optimal single-letter cost of the encoder is therefore given by
\begin{align}
    \Gamma_e^{\star}(R_1,R_2)=\inf_{\alpha,\beta} \Big\{\sum_{x^2}\lambda_{x^2}&c_1^x(q_{x^2}), \nonumber \\  &(\lambda_{x^2},q_{x^2})_{x^2} \in \argmin_{(\lambda_{x^2},q_{x^2})_{x^2}} C_2^{\star}(\alpha,\beta)\Big\}.
\end{align}

    %\centering
    %\includegraphics[width=10cm]{Topc1.jpg}
    %\captionof{figure}{\textit{Encoder's optimal strategy region $(\alpha,\beta)$.}}
    %\label{fig:my_label}

%    \centering
 %   \includegraphics[width=10cm]{PROFILEc1.jpg}
  %  \captionof{figure}{\textit{fig.5 Encoder's cost values with respect to $(\alpha,\beta)$. %Lowest at $2.2001=\Gamma_e^{\star}$ corresponds to $\alpha^{\star} =0.99$ and $\beta^{\star}=0.97$.
   % }}
    %\label{fig:my_label}

\bibliographystyle{IEEEtran}
\bibliography{biblio}

\begin{thebibliography}{10}
\providecommand{\url}[1]{#1}
\csname url@samestyle\endcsname
\providecommand{\newblock}{\relax}
\providecommand{\bibinfo}[2]{#2}
\providecommand{\BIBentrySTDinterwordspacing}{\spaceskip=0pt\relax}
\providecommand{\BIBentryALTinterwordstretchfactor}{4}
\providecommand{\BIBentryALTinterwordspacing}{\spaceskip=\fontdimen2\font plus
\BIBentryALTinterwordstretchfactor\fontdimen3\font minus
  \fontdimen4\font\relax}
\providecommand{\BIBforeignlanguage}[2]{{%
\expandafter\ifx\csname l@#1\endcsname\relax
\typeout{** WARNING: IEEEtran.bst: No hyphenation pattern has been}%
\typeout{** loaded for the language `#1'. Using the pattern for}%
\typeout{** the default language instead.}%
\else
\language=\csname l@#1\endcsname
\fi
#2}}
\providecommand{\BIBdecl}{\relax}
\BIBdecl

\bibitem{crawford1982}
V.~Crawford and J.~Sobel, ``Strategic information transmission,''
  \emph{Econometrica}, vol.~50, no.~6, pp. 1431--51, 1982.

\bibitem{KamenicaGentzkow11}
E.~Kamenica and M.~Gentzkow, ``Bayesian persuasion,'' \emph{American Economic
  Review}, vol. 101, pp. 2590 -- 2615, 2011.

\bibitem{jet}
M.~Le~Treust and T.~Tomala, ``Persuasion with limited communication capacity,''
  \emph{Journal of Economic Theory}, vol. 184, p. 104940, 2019.

\bibitem{cascade1}
H.~Yamamoto, ``Source coding theory for cascade and branching communication
  systems,'' \emph{IEEE Transactions on Information Theory}, vol.~27, no.~3,
  pp. 299--308, 1981.

\bibitem{cascade2}
\BIBentryALTinterwordspacing
D.~Vasudevan, C.~Tian, and S.~N. Diggavi, ``Lossy source coding for a cascade
  communication system with side-informations,'' 2006. [Online]. Available:
  \url{http://infoscience.epfl.ch/record/105103}
\BIBentrySTDinterwordspacing

\bibitem{koesslerlaclau}
F.~Koessler, M.~Laclau, and T.~Tomala, ``Interactive information design,''
  \emph{Mathematics of Operations Research}, 06 2021.

\bibitem{sar1}
S.~Saritas, S.~Yuksel, and S.~Gezici, ``Quadratic multi-dimensional signaling
  games and affine equilibria,'' \emph{IEEE Transactions on Automatic Control},
  vol.~62, no.~2, p. 605–619, Feb 2017.

\bibitem{SaritasFurrerGeziciLinderYukselISIT2019}
S.~{Sar{\i}ta\c{s}}, P.~{Furrer}, S.~{Gezici}, T.~{Linder}, and
  S.~{Y\"{u}ksel}, ``On the number of bins in equilibria for signaling games,''
  in \emph{2019 IEEE International Symposium on Information Theory (ISIT)},
  2019, pp. 972--976.

\bibitem{sar2}
S.~Sar\i{}ta\c{s}, S.~Y\"{u}ksel, and S.~Gezici, ``Dynamic signaling games with
  quadratic criteria under nash and stackelberg equilibria,''
  \emph{Automatica}, vol. 115, no.~C, May 2020.

\bibitem{dughmi}
S.~Dughmi, D.~Kempe, and R.~Qiang, ``Persuasion with limited communication,''
  in \emph{Proceedings of the 2016 ACM Conference on Economics and
  Computation}, ser. EC ’16.\hskip 1em plus 0.5em minus 0.4em\relax New York,
  NY, USA: Association for Computing Machinery, 2016, p. 663–680.

\bibitem{voraachievable}
A.~S. Vora and A.~A. Kulkarni, ``Achievable rates for strategic
  communication,'' in \emph{2020 IEEE International Symposium on Information
  Theory (ISIT)}, 2020, pp. 1379--1384.

\bibitem{vorakulkarniinformationextraction}
------, ``Information extraction from a strategic sender over a noisy
  channel,'' in \emph{2020 59th IEEE Conference on Decision and Control (CDC)},
  2020, pp. 354--359.

\bibitem{vorakulkarni2021optimalquestionnaires}
------, ``Optimal questionnaires for screening of strategic agents,'' in
  \emph{ICASSP 2021-2021 IEEE International Conference on Acoustics, Speech and
  Signal Processing (ICASSP)}.\hskip 1em plus 0.5em minus 0.4em\relax IEEE,
  2021, pp. 8173--8177.

\bibitem{AkyolLangbortBasar15}
E.~Akyol, C.~Langbort, and T.~Ba\c{s}ar, ``Strategic compression and
  transmission of information,'' in \emph{IEEE Information Theory Workshop -
  Fall (ITW)}, Oct 2015, pp. 219--223.

\bibitem{akyol2017information}
E.~Akyol, C.~Langbort, and T.~Ba{\c{s}}ar, ``Information-theoretic approach to
  strategic communication as a hierarchical game,'' \emph{Proceedings of the
  IEEE}, vol. 105, no.~2, pp. 205--218, 2017.

\bibitem{LeTreustTomala(Allerton)16}
M.~Le~Treust and T.~Tomala, ``Information design for strategic coordination of
  autonomous devices with non-aligned utilities,'' \emph{IEEE Proc. of the 54th
  Allerton conference, Monticello, Illinois}, pp. 233--242, 2016.

\bibitem{pointtopoint}
------, ``Point-to-point strategic communication,'' \emph{IEEE Information
  Theory Workshop}, 2020.

\bibitem{wyner-it-1976}
A.~D. Wyner and J.~Ziv, ``The rate-distortion function for source coding with
  side information at the decoder,'' \emph{IEEE Transactions on Information
  Theory}, vol.~22, no.~1, pp. 1--11, 1976.

\bibitem{akyol2016role}
E.~Akyol, C.~Langbort, and T.~Ba\c{s}ar, ``On the role of side information in
  strategic communication,'' in \emph{IEEE International Symposium on
  Information Theory (ISIT)}, July 2016, pp. 1626--1630.

\bibitem{corsica2020}
R.~{Bou Rouphael} and M.~{Le Treust}, ``Impact of private observation in
  bayesian persuasion,'' \emph{International Conference on NETwork Games
  COntrol and OPtimization NetGCoop}, Mar. 2020.

\bibitem{LeTreustTomalaISIT21}
M.~Le~Treust and T.~Tomala, ``Strategic communication with decoder side
  information,'' \emph{Information Symposium on Information Theory (ISIT)},
  2021.

\bibitem{elgamal}
A.~El~Gamal and Y.-H. Kim, \emph{Network information theory}.\hskip 1em plus
  0.5em minus 0.4em\relax Cambridge university press, 2011.

\end{thebibliography}

\end{document}